\documentclass[letterpaper,english,preprintnumbers,amsmath,amssymb,superscriptaddress,nofootinbib,prl]{revtex4}
\usepackage{mathpazo}
\usepackage[T1]{fontenc}
\usepackage[latin9]{inputenc}
\usepackage{array}
\usepackage{verbatim}
\usepackage{multirow}
\usepackage{amsmath}
\usepackage{amsthm}
\usepackage{amssymb}
\usepackage{graphicx}
\usepackage{setspace}

\makeatletter

\pdfpageheight\paperheight
\pdfpagewidth\paperwidth

\providecommand{\tabularnewline}{\\}

\@ifundefined{textcolor}{}
{%
 \definecolor{BLACK}{gray}{0}
 \definecolor{WHITE}{gray}{1}
 \definecolor{RED}{rgb}{1,0,0}
 \definecolor{GREEN}{rgb}{0,1,0}
 \definecolor{BLUE}{rgb}{0,0,1}
 \definecolor{CYAN}{cmyk}{1,0,0,0}
 \definecolor{MAGENTA}{cmyk}{0,1,0,0}
 \definecolor{YELLOW}{cmyk}{0,0,1,0}
}
  \theoremstyle{remark}
  \ifx\thechapter\undefined
    \newtheorem{rem}{\protect\remarkname}
  \else
    \newtheorem{rem}{\protect\remarkname}[chapter]
  \fi 
  \theoremstyle{remark}
  \newtheorem*{rem*}{\protect\remarkname}
  \theoremstyle{plain}
 \ifx\thechapter\undefined 
    \newtheorem{thm}{\protect\theoremname}
  \else
    \newtheorem{thm}{\protect\theoremname}[chapter]
  \fi
  \theoremstyle{plain}
 \ifx\thechapter\undefined 
    \newtheorem{lem}{\protect\lemmaname}
  \else
    \newtheorem{lem}{\protect\lemmaname}[chapter]
  \fi
  \theoremstyle{plain}
  \ifx\thechapter\undefined
    \newtheorem{conjecture}{\protect\conjecturename}
  \else
    \newtheorem{conjecture}{\protect\conjecturename}[chapter]
  \fi

\makeatother

\usepackage{babel}
  \providecommand{\conjecturename}{Conjecture}
  \providecommand{\lemmaname}{Lemma}
  \providecommand{\remarkname}{Remark}
  \providecommand{\theoremname}{Theorem}
\begin{document}

\title{Entanglement and correlation functions of the quantum Motzkin spin-chain}

\author{Ramis Movassagh }

\email{q.eigenman@gmail.com}

\selectlanguage{english}%

\affiliation{Department of Mathematics, IBM T. J . Watson Research Center,
Yorktown Heights, NY 10598}

\date{\today}

\maketitle
We present exact results on the exactly solvable spin chain of Bravyi
et al {[}Phys. Rev. Lett. 109, 207202 (2012){]}. This model is a spin
one chain and has a Hamiltonian that is local and translationally
invariant in the bulk. It has a unique (frustration free) ground state
with an energy gap that is polynomially small in the system's size
($2n$). The half-chain entanglement entropy of the ground state is
$\frac{1}{2}\log n+const.$ \cite{Movassagh2012_brackets}. Here we
first write the Hamiltonian in the standard spin-basis representation.
We prove that at zero temperature, the magnetization is along the
$z-$direction i.e., $\langle s^{x}\rangle=\langle s^{y}\rangle=0$
(everywhere on the chain). We then analytically calculate $\langle s^{z}\rangle$
and the two-point correlation functions of $s^{z}$. By analytically
diagonalizing the reduced density matrices, we calculate the Schmidt
rank, von Neumann and R\'enyi entanglement entropies for: 1. Any partition
of the chain into two pieces (not necessarily in the middle) and 2.
$L$ consecutive spins centered in the middle. Further, we identify
entanglement Hamiltonians (Eqs. \eqref{eq:Hcut} and \eqref{eq:H_L}).
We prove a small lemma (Lemma \eqref{lem:GeneralizedBallot}) on the
combinatorics of lattice paths using the reflection principle to relate
and calculate the Motzkin walk 'height' to spin expected values. We
also calculate the, closely related, (scaled) correlation functions
of Brownian excursions. The known features of this model are summarized
in a table in Sec. \ref{sec:Context-and-summary}. 
\tableofcontents{}
\section{\label{sec:Context-and-summary}Context and summary of the results}
\begin{comment}
The ground states of quantum many-body systems (QMBS) typically satisfy
an ``area law'', which says that the entanglement entropy between
a subsystem and the rest of the system is proportional to the area
of the boundary of the subsystem \cite{Eisert2010}. A generic state,
with probability one, violates the area law maximally \cite{hayden2006aspects}.
However, the ground states of QMBS are believed to be highly non-generic;
it is widely believed that for physically reasonable (as defined in
\cite{swingle2013universal}) and gapped QMBS, the area law holds.
Further, gapless systems such as Hamiltonians near the critical point
are believed to violate the area by at most a logarithmic factor in
the system's size \cite{Eisert2010,osterloh2002scaling,swingle2013universal}.
This believe by and large is an extrapolation from the theory of 1+1
dimensional quantum critical systems whose continuum limits are conformal
field theories (CFT). In such systems, there is rather a strong evidence
that the entanglement entropy of $L$ consecutive spins scales as
$\frac{c}{3}\log(L)$, where $c$ is the central charge \cite{holzhey1994geometric,VidalKitaev,Korepin2004,Cardy2009}.
\end{comment}

In recent times, existence and quantification of long range entanglement
in physical systems as a way of probing quantum phases of matter has
gained much attention \cite{facchi2015large,sachdev2007quantum,Korepin2004}.
Substantial amount of entanglement in the ground state may be utilized
to achieve quantum processing tasks such as spin state transfer \cite{pyrkov2014quantum}.
Most quantum interactions are local and a local Hamiltonian is frustration
free (FF) if the ground state is also the ground state of every one
of the terms in the interaction (summands). Such Hamiltonians afford
mathematical amenities that enable extraction of rich physics \cite{AKLT},
provide certain inherent stabilities against perturbations \cite{Kraus2008},
and their ground states can be engineered by dissipation \cite{Verstraete2009}. 

From a computer science and quantum complexity perspective, FF Hamiltonians
provide a natural bridge to physics where projectors that model the
local interactions are analogous to constraints in conventional satisfiability
problems \cite{gharibian2014quantum,Bravyi06}. It is interesting
to ask, how rich and entangled can FF quantum many-body systems be?
Much is known about local FF quantum spin$-1/2$ chains. For example
the ground state entanglement entropy is zero \cite{chen2011no} and
their energy gap has been classified \cite{Bravyi_Gosset2015}. In
general less is known for higher spin models. Recently it was shown
that in local FF systems, the gap $\Delta$ and the correlation length
$\xi$ are related by $\xi=\mathcal{O}\left(\Delta^{-1/2}\right)$
and that this bound is tight \cite{gosset2015correlation}. Moreover,
local generic FF spin chains with spin values $s\ge3/2$ are known
to have highly entangled and highly degenerate ground states \cite{Movassagh2010PRA}.
 It is natural then to investigate the properties of FF spin$-1$
chains. Well-known (non-critical) examples of these include the Heisenberg
ferromagnetic chain \cite{Koma95}, the AKLT model \cite{AKLT}, and
parent Hamiltonians of matrix product states \cite{Naechtaegale1992,cirac}.

Bravyi et al \cite{Movassagh2012_brackets} proposed a spin$-1$ model
that has a unique FF ground state, whose half-chain entanglement entropy
is $S=\frac{1}{2}\log n+c$, where $2n$ is the number of spins on
the chain. The Hamiltonian is local, translationally invariant in
the bulk, and has an energy gap to the first excited state that is
polynomially small in the size of the system. Despite the logarithmic
divergence of $S$ with $n$, this model was proved not to be described
by a conformal field theory (CFT) in the continuum limit \cite{movassagh2016supercritical}.\\

This work provides a more complete picture of this model. We take
the length of the chain to be $2n$; a chain with an odd number of
sites is done similarly. The way by which we take the limits and enforce
$L\ll n$ is explained in Sec. \ref{sub:Discussion-of-limits}. In
the following table, we exclude certain mathematical results of this
paper (e.g. Eqs. \eqref{eq:E_m2_Excur}--\eqref{eq:Connected_Physical})
as they are less relevant for the physics of the model. The table
below summarizes what is now known about this model with References
for finding the Results corresponding to any given Feature.\\
\begin{center}
\begin{tabular}{|c|c|c|}
\hline 
\multicolumn{1}{|c|}{\textsf{\textbf{Features}}} & \textsf{\textbf{Results}} & \textsf{\textbf{References}}\tabularnewline
\hline 
\hline 
\multirow{2}{*}{The Hamiltonian} & Local, translationally invariant in the bulk & \multirow{2}{*}{\cite{Movassagh2012_brackets}}\tabularnewline
 & Has boundary projectors, frustration free. & \tabularnewline
\hline 
The Hamiltonian in spin-representation & See the section & Eqs. \eqref{eq:HamiltonianSpinRep} -- \eqref{eq:BoundarySpinRep}\tabularnewline
Hamiltonian symmetry & $U(1)$ & Sec. (\ref{sub:Hamiltonian-in-spin-operator})\tabularnewline
\hline 
\multirow{2}{*}{Ground state is the Motzkin state: $|{\cal M}_{2n}\rangle$} & \multirow{2}{*}{Unique and frustration free} & \multirow{2}{*}{\cite{Movassagh2012_brackets}}\tabularnewline
 &  & \tabularnewline
\hline 
\multirow{2}{*}{The energy gap} & $\Theta\left(n^{-c}\right)$, $ $ Numerics indicate $c=3$  & \cite{Movassagh2012_brackets}\tabularnewline
 & Provably : $c\ge2$  &  \cite{movassagh2016supercritical}\tabularnewline
\hline 
Is the model describable by a (relativistic) CFT ? & No & \cite{movassagh2016supercritical}\tabularnewline
\hline 
\multirow{2}{*}{Expected Motzkin Height at $1<n_{1}<2n$} & \multirow{2}{*}{$\langle\widehat{m}_{n_{1}}\rangle=\frac{4}{\sqrt{3\pi}}\sqrt{n_{1}\left(1-n_{1}/2n\right)}$} & \multirow{2}{*}{Eqs. \eqref{eq:Magnetization} and \eqref{eq:E_m_Excur}}\tabularnewline
 &  & \tabularnewline
\hline 
Motzkin Height 2-point function with $1\ll L\ll n$ & $\langle\widehat{m}_{n-\frac{L}{2}}\widehat{m}_{n+\frac{L}{2}}\rangle=n-\frac{L}{3}+\frac{L^{2}}{4n}$ & Eq. \eqref{eq:2Pt_Final}\tabularnewline
\hline 
\multirow{2}{*}{Expected magnetization} & $\langle s^{x}\rangle=\langle s^{y}\rangle=0$ & Lemma \eqref{lem:SxSy}\tabularnewline
 & $\langle s_{n_{1}}^{z}\rangle=\frac{2}{\sqrt{3\pi}}\frac{\left(1-n_{1}/n\right)}{\sqrt{n_{1}\left(1-n_{1}/2n\right)}}$,
$\quad$ $1\ll n_{1}\ll2n$  & Eq. \eqref{eq:Sz_mean}\tabularnewline
\hline 
\multirow{2}{*}{$s^{z}$ two point function $1<n_{1}<n_{2}<2n$} & \multirow{2}{*}{$\langle s_{n_{1}}^{z}s_{n_{2}}^{z}\rangle=0$} & \multirow{2}{*}{Eqs. \eqref{eq:2pt_spinCorr} and \eqref{eq:SzSz_Excursion}}\tabularnewline
 &  & \tabularnewline
\hline 
Bipartite Schmidt rank & $\chi_{n_{1}}=\min\left\{ n_{1},2n-n_{1}\right\} +1$ & Eq. \eqref{eq:chi_n1}\tabularnewline
Bipartite von Neumann entropy about $1<n_{1}<2n$ & $S_{\mbox{cut}}=\frac{1}{2}\log\left[\frac{n_{1}\left(2n-n_{1}\right)}{n}\right]+const.$ & Eq. \eqref{eq:Scut}\tabularnewline
Bipartite R\'enyi entropy about $1<n_{1}<2n$ & $S_{\mbox{cut}}^{\kappa}=\frac{1}{2}\log\left[\frac{n_{1}\left(2n-n_{1}\right)}{n}\right]+f(\kappa)$ & Eq. \eqref{eq:Scut_Renyi}\tabularnewline
\hline 
Schmidt rank of $\mbox{ }1\ll L\ll n\mbox{ }$ middle spins & $\chi_{L}=2L+1$ & Eq. \eqref{eq:chi_L}\tabularnewline
von Neumann entropy of $\mbox{ }1\ll L\ll n\mbox{ }$ middle spins & $S_{L}=\frac{1}{2}\log(L)+const.$ & Eq. \eqref{eq:Entropy_Final}\tabularnewline
R\'enyi entropy of $\mbox{ }1\ll L\ll n\mbox{ }$ middle spins & $S_{L}^{\kappa}=\frac{1}{2}\log(L)+g(\kappa)$ & Eq. \eqref{eq:SL_Renyi}\tabularnewline
\hline 
\end{tabular}\\
\par\end{center}
\begin{rem}
The new results listed in the table above, which are analytically
derived below, give good agreements with numerical density matrix
renormalization group (DMRG) calculations for $2n=96$ already \cite{AdrianChat}. 
\end{rem}
\section{\label{sec:The-local-Hamiltonian}The ground state and the local
Hamiltonian }
\subsection{The unique ground state}
Let us consider a spin$-1$ chain of length $2n$. An odd size chain
is done similarly. It is convenient to label the $d=3$ spin states
by $\left\{ 0,u,d\right\} $ where $0$ means a flat step, $u$ means
a step up and $d$ a step down. A Motzkin walk on $2n$ steps is any
walk from coordinates $\left(x,y\right)=\left(1,0\right)$ to $\left(x,y\right)=\left(2n,0\right)$
where at any intermediate step the coordinates $(x,y)$ can only change
to $\left(x+1,y\right)$, $\left(x+1,y+1\right)$ or $\left(x+1,y-1\right)$
with the walk not passing below the x-axis, i.e., $y\ge0$ everywhere
on the walk. One makes the following identifications for the spin
states: $|d\rangle=|-1\rangle$, $|u\rangle=|+1\rangle$, and $|0\rangle$
is self-identified. 

The unique ground state is the Motzkin \textit{state,} which is defined
to be the uniform superposition of all Motzkin walks on $2n$ steps
\cite{Movassagh2012_brackets}. We denote the Motzkin state by $|{\cal M}_{2n}\rangle$,
which mathematically reads 
\begin{equation}
|{\cal M}_{2n}\rangle=\frac{1}{\sqrt{N}}\sum_{s\in\mbox{Motzkin}}|s\rangle\label{eq:GS}
\end{equation}
where $N$ is the total number of Motzkin walks on $2n$ steps. See
Fig. \eqref{fig:Examples-of-MotzkinState} for examples of the Motzkin
State $|{\cal M}_{2n}\rangle$. 
\begin{figure}
\begin{centering}
\includegraphics[scale=0.38]{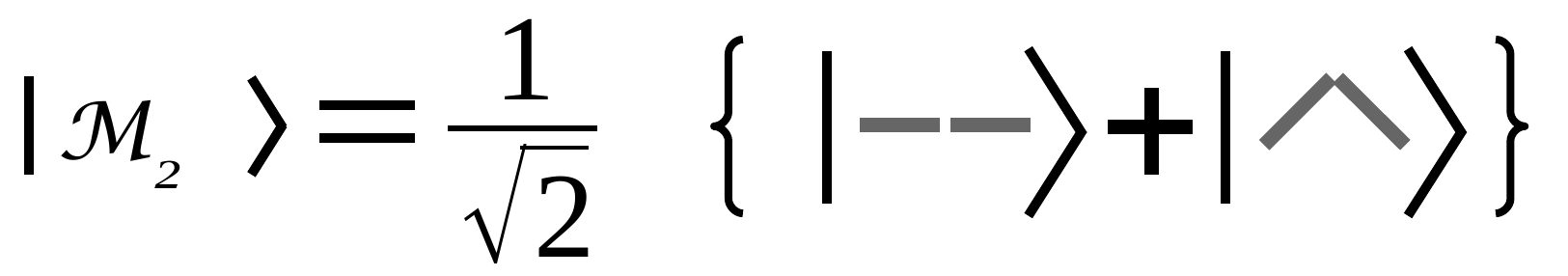}
\par\end{centering}
\begin{centering}
\includegraphics[scale=0.35]{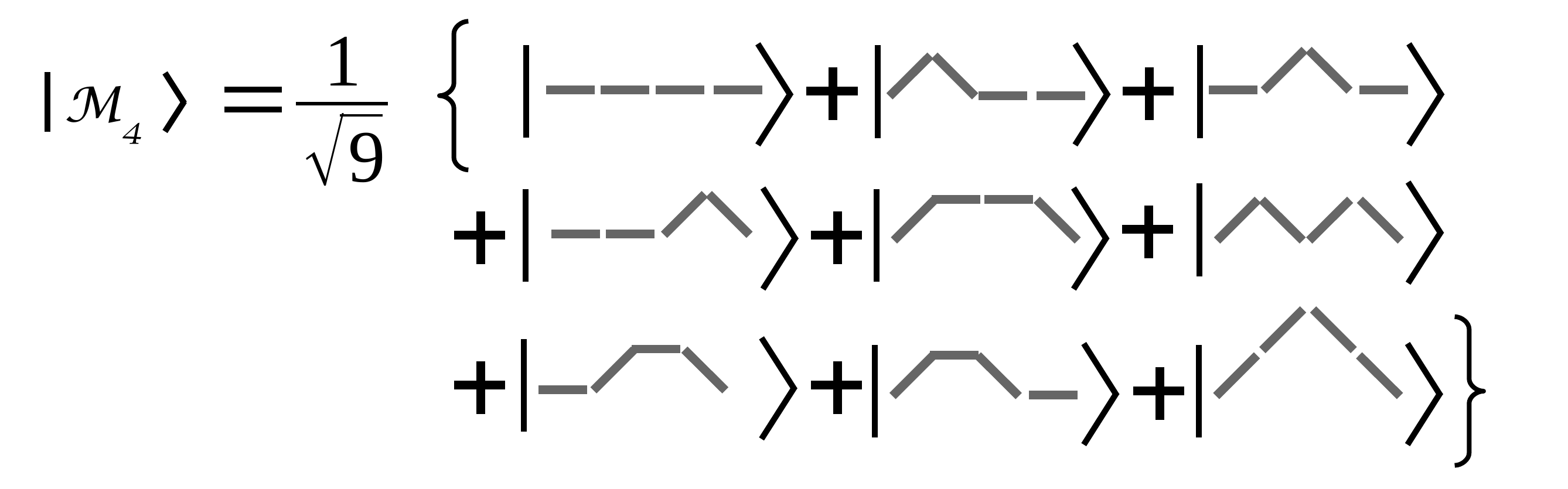}\caption{\label{fig:Examples-of-MotzkinState}Examples of the ground state
(Motzkin State) for two chains of different size.}
\par\end{centering}
\end{figure}
\subsection{\label{sub:Hamiltonian-in-spin-operator}Hamiltonian in spin-operator
representation and its symmetry}
In \cite{Movassagh2012_brackets}, it was shown that $|{\cal M}_{2n}\rangle$
is the unique and frustration free ground state of a local Hamiltonian,
which we now review. Consider the following local operations on any
Motzkin walk: interchanging zero with a non-flat step (i.e., $0d\leftrightarrow d0$
or $0u\leftrightarrow u0$) or interchanging a consecutive pair of
zeros with a peak (i.e., $00\leftrightarrow ud$). Any Motzkin walk
can be obtained from another one by a sequence of these local changes.
To construct a local Hamiltonian with projectors as interactions that
has the uniform superposition of the Motzkin walks as its zero energy
ground state, each of the local terms of the Hamiltonian has to annihilate
states that are symmetric under these interchanges.

The local Hamiltonian is \cite{Movassagh2012_brackets}
\begin{equation}
H=\sum_{j=1}^{2n-1}\Pi_{j,j+1}+\Pi_{boundary},\label{eq:H}
\end{equation}
where $\Pi_{j,j+1}$ implements the local changes discussed above
and is defined by 
\[
\Pi_{j,j+1}\equiv|D\rangle_{j,j+1}\langle D|+|U\rangle_{j,j+1}\langle U|+|\varphi\rangle_{j,j+1}\langle\varphi|
\]
where $|D\rangle=\frac{1}{\sqrt{2}}\left[|0d\rangle-|d0\rangle\right]$,
$|U\rangle=\frac{1}{\sqrt{2}}\left[|0u\rangle-|u0\rangle\right]$
and $|\varphi\rangle=\frac{1}{\sqrt{2}}\left[00\rangle-\mbox{ }|ud\rangle\right]$.
The boundary term $\Pi_{boundary}\equiv\left[|d\rangle_{1}\langle d|+|u\rangle_{2n}\langle u|\right]$
selects out the Motzkin state as the only ground state.

In the spin operator representation, taking $\hbar=1$ \cite{landau2013quantum}
\footnote{Only in this section $m\in\{-1,0,+1\}$ denotes the spin state; it
has nothing to do with the height of a walk that appears in the other
(sub)sections.}:
\begin{eqnarray*}
s^{2}|s,m\rangle & = & s\left(s+1\right)\mbox{ }|s,m\rangle\\
s^{z}|s,m\rangle & = & m\mbox{ }|s,m\rangle\\
S^{\pm}|s,m\rangle & = & \sqrt{s\left(s+1\right)-m\left(m\pm1\right)}\mbox{ }|s,m\pm1\rangle
\end{eqnarray*}
where $S^{\pm}=s^{x}\pm is^{y}$. Below we drop $s$-dependence as
it is always equal to one, and denote the ket $|s,m\rangle$ simply
by $|m\rangle$. For example the state of two consecutive spins being
$|0,1\rangle_{j,j+1}$ really means $|m=0\rangle_{j}\otimes|m=1\rangle_{j+1}$. We identify
\[
|u\rangle\equiv|+1\rangle=\left(\begin{array}{c}
1\\
0\\
0
\end{array}\right),\quad|0\rangle=\left(\begin{array}{c}
0\\
1\\
0
\end{array}\right),\quad|d\rangle\equiv|-1\rangle=\left(\begin{array}{c}
0\\
0\\
1
\end{array}\right)\quad.
\]
In this basis \cite{landau2013quantum} 
\[
s^{x}=\frac{1}{\sqrt{2}}\left(\begin{array}{ccc}
0 & 1 & 0\\
1 & 0 & 1\\
0 & 1 & 0
\end{array}\right),\quad s^{y}=\frac{1}{\sqrt{2}}\left(\begin{array}{ccc}
0 & -i & 0\\
i & 0 & -i\\
0 & i & 0
\end{array}\right),\quad s^{z}=\left(\begin{array}{ccc}
1 & 0 & 0\\
0 & 0 & 0\\
0 & 0 & -1
\end{array}\right),
\]
and
\begin{align}
s^{x}\mbox{ }|0\rangle & =\frac{1}{\sqrt{2}}\left[|-1\rangle+|+1\rangle\right],\quad s^{x}|+1\rangle=\frac{1}{\sqrt{2}}|0\rangle,\quad s^{x}|-1\rangle=\frac{1}{\sqrt{2}}|0\rangle\mbox{ },\label{eq:Sx}\\
s^{y}\mbox{ }|0\rangle & =\frac{i}{\sqrt{2}}\left[|-1\rangle-|+1\rangle\right],\quad s^{y}|+1\rangle=\frac{i}{\sqrt{2}}|0\rangle,\quad s^{y}|-1\rangle=\frac{-i}{\sqrt{2}}|0\rangle\mbox{ },\label{eq:Sy}\\
s^{z}\mbox{ }|0\rangle & =0,\quad s^{z}\mbox{ }|+1\rangle=|+1\rangle,\quad s^{z}\mbox{ }|-1\rangle=-|-1\rangle\mbox{ }.\label{eq:Sz}
\end{align}
The local terms on any pair of nearest neighbor spins become
\begin{eqnarray*}
|U\rangle\langle U| & = & \frac{1}{2}\left\{ \left(|0,1\rangle-|1,0\rangle\right)\left(\langle0,1|-\langle1,0|\right)\right\} \\
 & = & \frac{1}{2}\left\{ \left(|0\rangle\langle0|\otimes|1\rangle\langle1|-|1\rangle\langle0|\otimes|0\rangle\langle1|-|0\rangle\langle1|\otimes|1\rangle\langle0|+|1\rangle\langle1|\otimes|0\rangle\langle0|\right)\right\} ,\\
|D\rangle\langle D| & = & \frac{1}{2}\left\{ \left(|0,-1\rangle-|-1,0\rangle\right)\left(\langle0,-1|-\langle-1,0|\right)\right\} \\
 & = & \frac{1}{2}\left\{ \left(|0\rangle\langle0|\otimes|-1\rangle\langle-1|-|-1\rangle\langle0|\otimes|0\rangle\langle-1|-|0\rangle\langle-1|\otimes|-1\rangle\langle0|+|-1\rangle\langle-1|\otimes|0\rangle\langle0|\right)\right\} \\
|\varphi\rangle\langle\varphi| & = & \frac{1}{2}\left\{ \left(|00\rangle-|1,-1\rangle\right)\left(\langle00|-\langle1,-1|\right)\right\} \\
 & = & \frac{1}{2}\left\{ \left(|0\rangle\langle0|\otimes|0\rangle\langle0|-|1\rangle\langle0|\otimes|-1\rangle\langle0|-|0\rangle\langle1|\otimes|0\rangle\langle-1|+|1\rangle\langle1|\otimes|-1\rangle\langle-1|\right)\right\} .
\end{eqnarray*}
Further simplification gives $\Pi_{j,j+1}$
\begin{eqnarray}
\Pi_{j,j+1} & = & \frac{1}{2}\left\{ |0\rangle_{j}\langle0|\otimes\mathbb{I}_{j+1}+|1\rangle_{j}\langle1|\otimes\left\{ |0\rangle_{j+1}\langle0|+|-1\rangle_{j+1}\langle-1|\right\} +|-1\rangle_{j}\langle-1|\otimes|0\rangle_{j+1}\langle0|\right\} \label{eq:HamiltonianSpinRep}\\
 & - & \frac{1}{2}\left\{ |1\rangle_{j}\langle0|\otimes|0\rangle_{j+1}\langle1|+|-1\rangle_{j}\langle0|\otimes|0\rangle_{j+1}\langle-1|+|1\rangle_{j}\langle0|\otimes|-1\rangle_{j+1}\langle0|+h.c.\right\} \nonumber 
\end{eqnarray}
\begin{singlespace}
From the action of the spin operators on the states, it is easy to
see that 
\[
\begin{array}{ccccccc}
|0\rangle\langle0| & = & \mathbb{I}-\left(s^{z}\right)^{2} & \quad & |-1\rangle\langle-1| & = & \frac{1}{2}\left(\mathbb{I}-s^{z}\right)s^{z}\\
|1\rangle\langle1| & = & \frac{1}{2}\left(\mathbb{I}+s^{z}\right)s^{z} &  & |0\rangle\langle-1| & = & \frac{1}{2\sqrt{2}}S^{+}\left(\mathbb{I}-s^{z}\right)s^{z}\\
|0\rangle\langle1| & = & \frac{1}{2\sqrt{2}}S^{-}\left(\mathbb{I}+s^{z}\right)s^{z} &  & |-1\rangle\langle0| & = & \frac{1}{2\sqrt{2}}s^{z}\left(\mathbb{I}-s^{z}\right)S^{-}\\
|1\rangle\langle0| & = & \frac{1}{2\sqrt{2}}s^{z}\left(\mathbb{I}+s^{z}\right)S^{+}
\end{array}
\]
These can be plugged into the right hand side of the equation above
to obtain $\Pi_{j,j+1}$. The boundary projector reads
\begin{equation}
\Pi_{boundary}=\frac{1}{2}\left(\mathbb{I}_{1}-s_{1}^{z}\right)s_{1}^{z}+\frac{1}{2}\left(\mathbb{I}_{2n}+s_{2n}^{z}\right)s_{2n}^{z}.\label{eq:BoundarySpinRep}
\end{equation}
Eqs. \eqref{eq:HamiltonianSpinRep}--\eqref{eq:BoundarySpinRep} can now
be plugged into Eq. \eqref{eq:H} to fully express it in terms of spin
operators.
\end{singlespace}
\begin{rem*}
\begin{singlespace}
The model has a $U(1)$ symmetry. Moreover, $\widehat{m}_{2n}\equiv\sum_{j=1}^{2n}s_{j}^{z}$
commutes with the Hamiltonian and is therefore a conserved quantity.\end{singlespace}
\end{rem*}
\section{Correlation functions}
Let $n_{1}$ and $n_{2}$ be two sites on the chain such that $1<n_{1}<n_{2}<2n$.
At zero temperature, one defines the correlations in the ground state
by
\begin{eqnarray*}
\langle s_{n_{1}}^{z}\rangle & \equiv & \langle{\cal M}_{2n}|s_{n_{1}}^{z}\mbox{ }|{\cal M}_{2n}\rangle\\
\langle s_{n_{1}}^{z}s_{n_{2}}^{z}\rangle & \equiv & \langle{\cal M}_{2n}|s_{n_{1}}^{z}s_{n_{2}}^{z}\mbox{ }|{\cal M}_{2n}\rangle
\end{eqnarray*}
where $s^{z}$ on any site can be written in terms of $|u\rangle$,
$|d\rangle$ and $|0\rangle$ by $s^{z}=\mbox{ }|u\rangle\langle u|\mbox{ }-\mbox{ }|d\rangle\langle d|$.
\begin{figure}
\centering{}\includegraphics[scale=0.4]{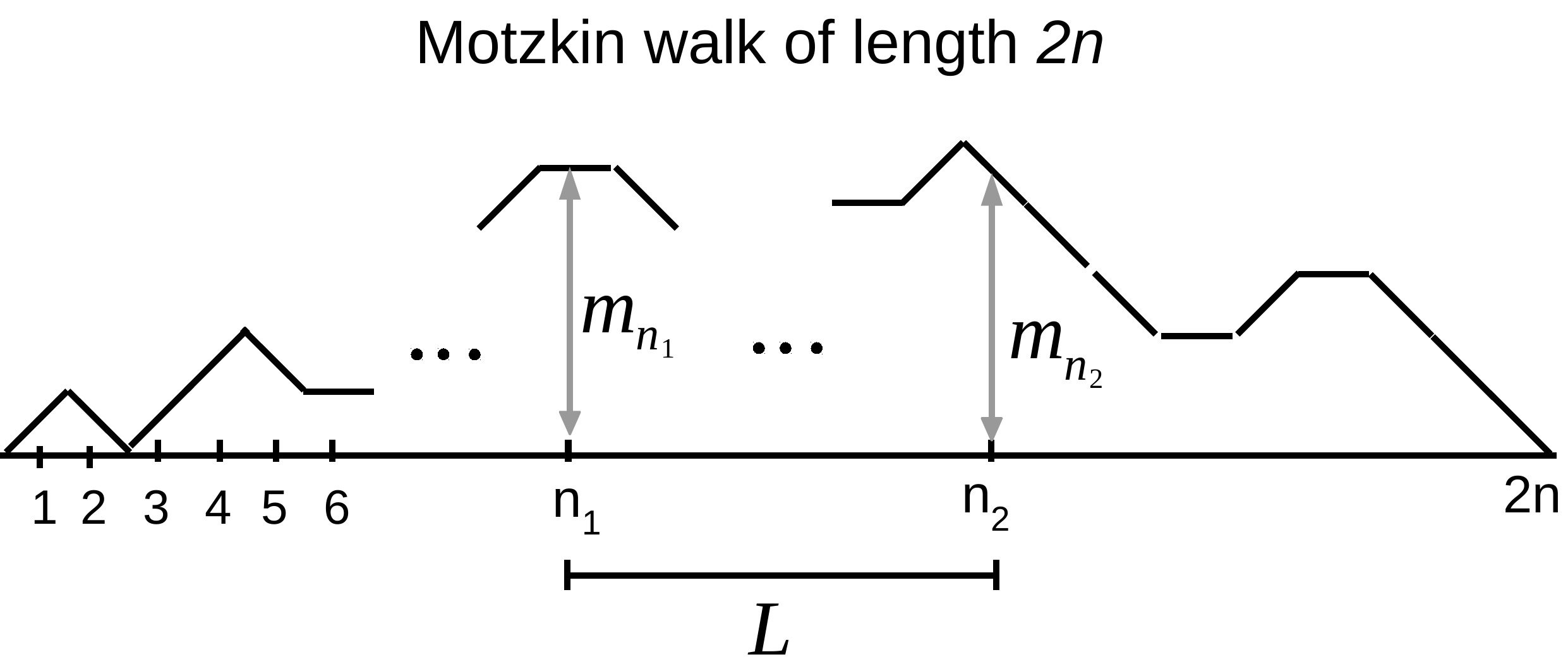}\caption{\label{fig:Geometry-of-the}Geometry of a Motzkin walk and height
correlations.}
\end{figure}
Natural and related combinatorial problems are the height and height-height
correlation functions. Suppose the height operator at the site $n_{1}$
is $\widehat{m}_{n_{1}}$ and at $n_{2}$ is $\widehat{m}_{n_{2}}$.
One defines these correlation functions by
\begin{eqnarray*}
\langle\widehat{m}_{n_{1}}\rangle & \equiv & \langle{\cal M}_{2n}|\widehat{m}_{n_{1}}|{\cal M}_{2n}\rangle\\
\langle\widehat{m}_{n_{1}}\widehat{m}_{n_{2}}\rangle & \equiv & \langle{\cal M}_{2n}|\widehat{m}_{n_{1}}\widehat{m}_{n_{2}}|{\cal M}_{2n}\rangle\mbox{ },
\end{eqnarray*}
where
\begin{equation}
\widehat{m}_{n_{1}}=\sum_{j=1}^{n_{1}}s_{j}^{z}\mbox{ },\qquad\widehat{m}_{n_{2}}=\sum_{j=1}^{n_{2}}s_{j}^{z}\mbox{ },\label{eq:mTozMap}
\end{equation}
which relate heights to $s^{z}$. See Fig. \eqref{fig:Geometry-of-the}
for the geometry. Since $\widehat{m}_{n_{1}}-\widehat{m}_{(n_{1}-1)}=s_{n_{1}}^{z}$,
the height correlation functions are the generators of the spin correlation
functions. Similarly, the expected bivariate difference equation is the two point function. Below we calculate the exact values of $\langle s_{n_{1}}^{z}\rangle $ and $\langle s_{n_{1}}^{z}s_{n_{2}}^{z}\rangle $ and show that the spin expected values can be obtained from the height expectations by differentiations  $\Theta(1/n)$ errors (see Eqs. \eqref{eq:Magnetization2}-\eqref{eq:Sz_mean}, and  \eqref{eq:bivariateTwopoint}-\eqref{eq:2pt_spinCorr}):
\begin{eqnarray*}
\langle s_{n_{1}}^{z}\rangle & = & \frac{\partial\langle\widehat{m}_{n_{1}}\rangle}{\partial n_{1}}+\Theta(1/n),\\
\langle s_{n_{1}}^{z}s_{n_{2}}^{z}\rangle & = & \frac{\partial^{2}\langle\widehat{m}_{n_{1}}\widehat{m}_{n_{2}}\rangle}{\partial n_{1}\mbox{ }\partial n_{2}}+\Theta(1/n).
\end{eqnarray*}
\begin{rem}
It would be interesting to find that $\langle s_{n_{1}}^{z}\rangle\ne0$.
In particular, we have that $\langle s_{1}^{z}\rangle>0$ and $\langle s_{2n}^{z}\rangle<0$
because of the non-negative constraint imposed by Motzkin walks.
\end{rem}
\begin{rem}
In physics one usually expects $\langle s_{n_{1}}^{z}s_{n_{2}}^{z}\rangle=G(|n_{2}-n_{1}|).$
Moreover, it is expected that in the asymptotic limit $G(|n_{2}-n_{1}|)\sim|n_{2}-n_{1}|^{\theta}$.
In Sec. \ref{sub:Two-point-functions} we shall see that $\theta=0$. 
\end{rem}
The thermal one-point and two-point correlation functions are 
\begin{eqnarray*}
\langle s_{n_{1}}^{z}\rangle_{T} & \equiv & \frac{1}{Z(\beta)}\mbox{Tr}(s_{n_{1}}^{z}\mbox{ }e^{-\beta H})=\frac{1}{Z(\beta)}\sum_{\alpha}\langle\alpha|s_{n_{1}}^{z}\mbox{ }e^{-\beta E_{\alpha}}|\alpha\rangle\\
\langle s_{n_{1}}^{z}s_{n_{2}}^{z}\rangle_{T} & \equiv & \frac{1}{Z(\beta)}\mbox{Tr}(s_{n_{1}}^{z}s_{n_{2}}^{z}\mbox{ }e^{-\beta H})=\frac{1}{Z(\beta)}\sum_{\alpha}\langle\alpha|s_{n_{1}}^{z}s_{n_{2}}^{z}\mbox{ }e^{-\beta E_{\alpha}}|\alpha\rangle
\end{eqnarray*}
where $Z(\beta)=\mbox{Tr}(e^{-\beta H})$ is the partition function.
In the presence of an external field $h$, the partition function
is also a function of the field i.e., $Z(\beta,h)$. We currently
do not have a good enough understanding of the spectrum above the
ground state to make analytical progress on this; $\langle s_{n_{1}}^{z}\rangle_{T}$
and $\langle s_{n_{1}}^{z}s_{n_{2}}^{z}\rangle_{T}$ need to be calculated
numerically. 
\subsection{\label{sub:Discussion-of-limits}Limits with respect to large parameters:
Physical vs. Excursions}
In Sec. \ref{sub:Two-point-functions} we will calculate two-point
correlation functions, and in Sec.\ref{sub:Block-Entanglement-Entropy}
the block entanglement entropy of the $L-$middle consecutive spins
denoted by $S_{L}$ (where $L\equiv n_{2}-n_{1}$). We are interested
in the asymptotic form and scaling of these quantities with respect
to $n$ and $L$. 

We have two large parameters, one is $2n$ which is the size of the
chain and the other is $L$, which is the number of consecutive spins
centered about the middle of the chain. In the derivation of the two-point
function (Sec. \ref{sub:Two-point-functions}) and block entanglement
entropy (Sec. \ref{sub:Block-Entanglement-Entropy}) care must be
taken in taking the limits. Two ways of taking the limits that we
like to concern ourselves with are what we call Physical and Excursions:
\begin{enumerate}
\item \textbf{Physical: }In this limit, one first takes the limit of the
system size to infinity while keeping $L$ \textit{fixed}. Once the
asymptotic with respect to $n$ is obtained, one then assumes a large
$L$ and derives asymptotic results. This corresponds to taking the
thermodynamical limit in physics. Mathematically, the ``Physical''
limits, involving $n$ and $L$, that we shall derive below are:
\[
\begin{array}{c}
\lim_{L\rightarrow\infty}\left\{ \lim_{n\rightarrow\infty}\langle\widehat{m}_{n_{1}}\widehat{m}_{n_{2}}\rangle\right\} \\
\lim_{L\rightarrow\infty}\left\{ \lim_{n\rightarrow\infty}S_{L}\right\} 
\end{array}
\]
where $\langle\widehat{m}_{n_{1}}\widehat{m}_{n_{2}}\rangle$ and
$S_{L}$ are functions of $n$ and $L$. In practice, however, when
one makes plots of such asymptotically obtained results or when one
runs numerical algorithms such as DMRG, both $n$ and $L$ are finite,
and care must be taken as to what ratios $\frac{L}{n}$ are small
enough to be compared with analytical formulas. 
\item \textbf{Brownian Excursions:} A different asymptotic can be obtained where
$n_{1}=2\lambda n$ and $n_{2}=2\mu n$ with $0<\lambda<\mu<1$. In
this limit, $L\equiv n_{2}-n_{1}$ tends to infinity simultaneously
with $n$, and results from universal convergence of random walks
to Brownian excursions can be evoked to calculate the scaling of the
two point function. \end{enumerate}
\begin{rem}
In calculations of entanglement, we satisfy ourselves with the physics
of the model and leave derivations in the Excursion limit for future
work.\\
\end{rem}
\section{Height and height-height correlation functions}
In order to calculate correlation functions, we first prove the following
lemma:
\begin{lem}
\label{lem:GeneralizedBallot}Let $D_{L,m_{1},m_{2}}$ be the number
of non-negative walks on $L$ steps that connect the points $(0,m_{1})$
and $(L,m_{2})$ where in each intermediate step the coordinates $(x,y)$
can change to either $(x+1,y+1)$ or $(x+1,y-1)$. In other words,
these are just like Dyck paths \cite{StanleyVol2}, except that they
start and end at heights $m_{1}\ge0$ and $m_{2}\ge0$ respectively.
This number is zero if $|m_{2}-m_{1}|>L$ or if $m_{2}-m_{1}\ne L$
mod ($2$). Otherwise it is given by 
\begin{equation}
D_{L,m_{1},m_{2}}=\left(\begin{array}{c}
L\\
\frac{L+|m_{2}-m_{1}|}{2}
\end{array}\right)-\left(\begin{array}{c}
L\\
\frac{L+(m_{2}+m_{1})}{2}+1
\end{array}\right)\label{eq:Eq_Theorem}
\end{equation}
\end{lem}
\begin{proof}
We prove this by first counting the total number of paths that connect
$(x,y)=(0,m_{1})$ and $(x,y)=(L,m_{2})$ and then subtract from it
the total number of paths that become negative. To count the latter
we use the reflection principle. Suppose for now that $m_{2}\ge m_{1}$,
each path that connects $(0,m_{1})$ and $(L,m_{2})$ necessarily
has $(m_{2}-m_{1})$ excess number of step ups. Consequently, the
total number of step downs are $\frac{L-(m_{2}-m_{1})}{2}$. Therefore,
the total number of paths that connect $(0,m_{1})$ and $(L,m_{2})$
is $\left(\begin{array}{c}
L\\
\frac{L-(m_{2}-m_{1})}{2}
\end{array}\right)$. Had it been that $m_{1}\ge m_{2}$ this number clearly would be
$\left(\begin{array}{c}
L\\
\frac{L-(m_{1}-m_{2})}{2}
\end{array}\right)$. Using the fact that $\left(\begin{array}{c}
L\\
m
\end{array}\right)=\left(\begin{array}{c}
L\\
L-m
\end{array}\right)$ we arrive at $\left(\begin{array}{c}
L\\
\frac{L+|m_{2}-m_{1}|}{2}
\end{array}\right)$.
\begin{figure}
\centering{}\includegraphics[scale=0.28]{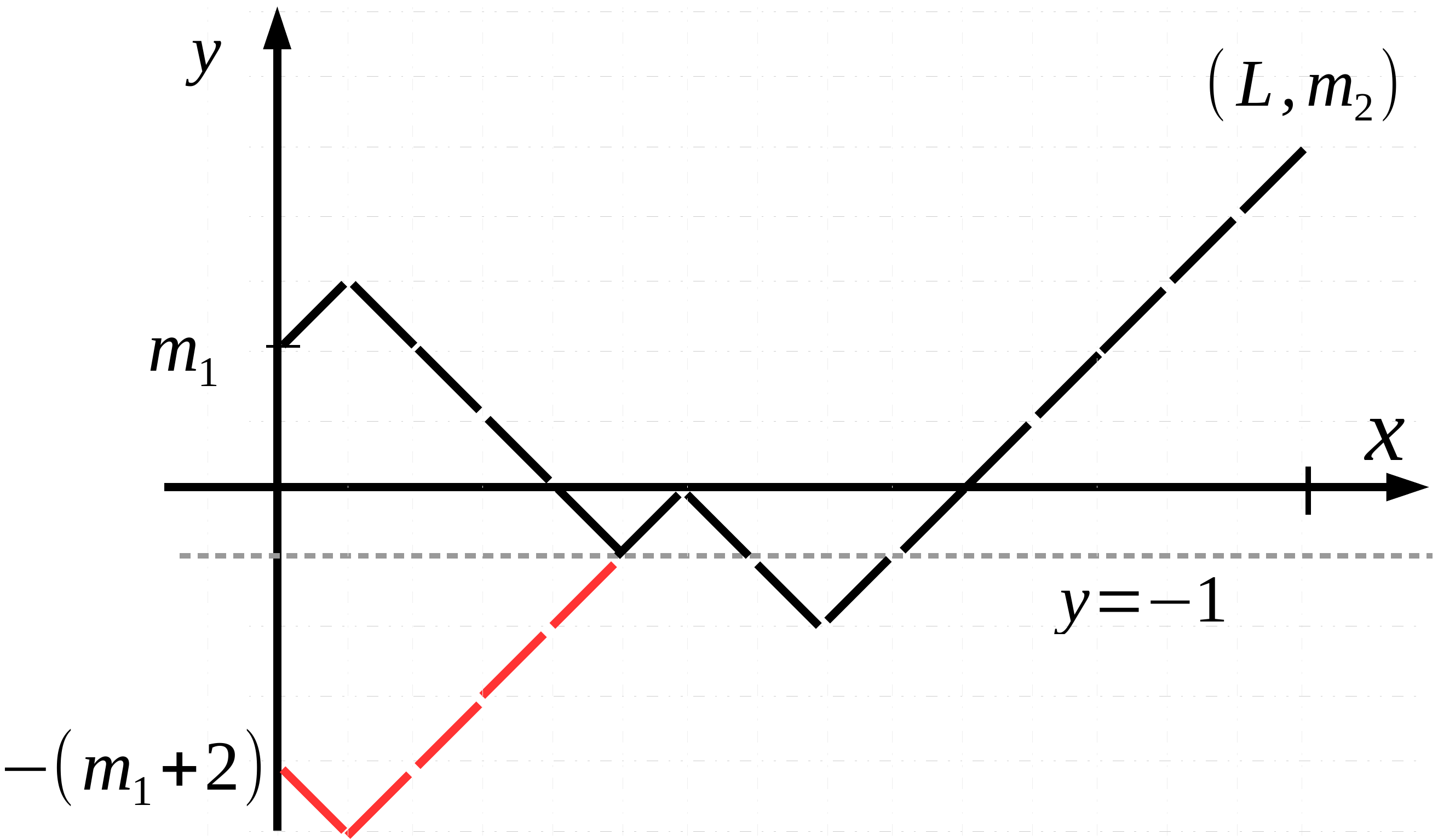}\caption{\label{fig:Bijection-between-paths}Bijection between paths starting
at $(0,m_{1})$ that become negative (i.e., ``bad'' paths) and paths
that start at $(0,-(m_{1}+2))$.}
\end{figure}
We want to subtract the number of ``bad'' paths, which cross the
$y=0$ line at least once. Any bad path, $P$, necessarily has to
reach the line $y=-1$ (see Fig. \eqref{fig:Bijection-between-paths}).
Define a new path by reflecting about $y=-1$ the part of $P$ up
to the first point it touches the line. To every such ``bad'' path
$P$, there corresponds a path $P'$ between $(0,-(m_{1}+2))$ and
$(L,m_{2})$. Moreover, every path between $(0,-(m_{1}+2))$ and $(L,m_{2})$
necessarily crosses $y=-1$ line and by reflection across it will
be mapped to a ``bad'' path. Therefore there is a bijection between
all the ``bad'' paths and the lattice paths that connect $(0,-(m+2))$
and $(L,m_{2})$. The total number of excess step ups are $m_{1}+m_{2}+2$
and hence there are $\frac{L-(m_{1}+m_{2}+2)}{2}$ step downs. The
total number of bad paths is then $\left(\begin{array}{c}
L\\
\frac{L-(m_{1}+m_{2})}{2}-1
\end{array}\right)$. Subtracting this from the total number of paths we obtain $D_{L,m_{1},m_{2}}$
(Eq.\ref{eq:Eq_Theorem}). \end{proof}
\begin{rem}
When $m_{1}=m_{2}=0$, and $L$ is even, $D_{L,0,0}=\frac{1}{\frac{L}{2}+1}\left(\begin{array}{c}
L\\
L/2
\end{array}\right)$ which is the $(L/2)^{th}$ Catalan number. When $m_{1}=0$ and $m_{2}=m$,
$D_{L,0,m}=\frac{m+1}{L+1}\left(\begin{array}{c}
L+1\\
\frac{L-m}{2}
\end{array}\right)$ is the solution of a variation of the Ballot problem, where ties
are allowed. 
\end{rem}
The number of walks between $(n_{1},m_{1})$ and $(n_{2},m_{2})$
made up of up and down steps, as well as, flat steps (Motzkin like
walks) is (recall that $L\equiv n_{2}-n_{1}$) 
\begin{eqnarray}
M_{L,m_{1},m_{2}} & = & \sum_{k=0}^{L-|m_{2}-m_{1}|}\left(\begin{array}{c}
L\\
k
\end{array}\right)D_{L-k,m_{1},m_{2}}=\sum_{k=0}^{L-|m_{2}-m_{1}|}\left(\begin{array}{c}
L\\
k
\end{array}\right)\left\{ \left(\begin{array}{c}
L-k\\
\frac{L-k+|m_{2}-m_{1}|}{2}
\end{array}\right)-\left(\begin{array}{c}
L-k\\
\frac{L-k+(m_{2}+m_{1})}{2}+1
\end{array}\right)\right\} \label{eq:M_Ls}
\end{eqnarray}
Let  $2i=L-k-|m_{2}-m_{1}|$ to take care of the parity; the summand
after this transformation becomes
\begin{eqnarray}
M_{L,m_{1},m_{2},i} & = & E_{L,m_{1},m_{2},i}-F_{L,m_{1},m_{2},i},\label{eq:MLms}
\end{eqnarray}
where
\begin{eqnarray*}
E_{L,m_{1},m_{2},i} & \equiv & \left(\begin{array}{c}
L\\
2i+|m_{2}-m_{1}|
\end{array}\right)\left(\begin{array}{c}
2i+|m_{2}-m_{1}|\\
i+|m_{2}-m_{1}|
\end{array}\right)\\
F_{L,m_{1},m_{2},i} & \equiv & \left(\begin{array}{c}
L\\
2i+|m_{2}-m_{1}|
\end{array}\right)\left(\begin{array}{c}
2i+|m_{2}-m_{1}|\\
i+\frac{|m_{2}-m_{1}|+(m_{1}+m_{2})}{2}+1
\end{array}\right).
\end{eqnarray*}

The sum in Eq. \eqref{eq:M_Ls} reads 
\begin{equation}
M_{L,m_{1},m_{2}}=\sum_{i=0}^{\frac{L-|m_{2}-m_{1}|}{2}}E_{L,m_{1},m_{2},i}-F_{L,m_{1},m_{2},i}\quad.\label{eq:M_exact}
\end{equation}
Now using multinomial identities, and recalling the definition of
trinomial coefficients $\left(\begin{array}{ccccc}
 &  & L\\
x & , & y & , & z
\end{array}\right)\equiv\frac{L!}{x!y!z!}$ with $x+y+z=L$, we find 
\begin{eqnarray*}
E_{L,m_{1},m_{2},i} & = & {\scriptstyle \left(\begin{array}{ccccc}
 &  & L\\
L-2i-|m_{2}-m_{1}| & , & i & , & i+|m_{2}-m_{1}|
\end{array}\right)}\\
F_{L,m_{1},m_{2},i} & = & {\scriptstyle \left(\begin{array}{ccccc}
 &  & L\\
L-2i-|m_{2}-m_{1}| & , & i+\frac{|m_{2}-m_{1}|+(m_{1}+m_{2})}{2}+1 & , & i+\frac{|m_{2}-m_{1}|-(m_{1}+m_{2})}{2}-1
\end{array}\right)}.
\end{eqnarray*}
Comment: Recall that $C_{i}\equiv\frac{1}{i+1}\left(\begin{array}{c}
2i\\
i
\end{array}\right)$ be the $i^{th}$ Catalan number. As a special case, we see that from
Eq. \eqref{eq:M_exact} we have  $N=M_{2n,0,0}=\sum_{i=0}^{n}\left(\begin{array}{c}
2n\\
2i
\end{array}\right)C_{i}=\frac{1}{2n+1}\sum_{i=0}^{n}\left(\begin{array}{c}
2n+1\\
2n-2i\mbox{ },\mbox{ }i\mbox{ },\mbox{ }i+1
\end{array}\right)$.

\begin{comment}
We now derive a machinery that enables us to obtain the asymptotic
form of $M_{b,0,m}$, $M_{b,(m+p),0}$ and $M_{L,m,(m+p)}$-- the
latter will be used in the calculation of the 2-point function. Using
a refinement of Stirling's approximation (see for example \cite{movassagh2016supercritical})
we have for $x+y+z=0$ and ${\cal N}\gg1$ 
\[
\left(\begin{array}{ccccc}
 &  & {\cal N}\\
\frac{{\cal N}}{3}+z & , & \frac{{\cal N}}{3}+y & , & \frac{{\cal N}}{3}+x
\end{array}\right)\approx\frac{3^{{\cal N}+3/2}}{2\pi{\cal N}}\exp\left\{ -\frac{3}{2{\cal N}}\left(x^{2}+y^{2}+z^{2}\right)\right\} .
\]
\end{comment}
Using Eq. \eqref{eq:MLms} and \eqref{eq:M_exact} we obtain 
\begin{eqnarray}
M_{L,0,m}=\frac{m+1}{L+1}\sum_{i\ge0}\left(\begin{array}{c}
L+1\\
L-2i-m\mbox{ }\mbox{ },\mbox{ }i\mbox{ },\mbox{ }\mbox{ }i+m+1
\end{array}\right) \label{eq:Mnms}
\end{eqnarray}
\subsection{Asymptotic Analysis}
In the following sections we encounter sums (e.g., Eqs \eqref{eq:Mnms}) whose asymptotic values are desired. In what follows we will make extensive use of the Stirlings formula, as well as, the integral test in the theory of sequences and series and Euler-Maclaurin formula. The latter ensure the accuracy and convergence of the sums to the obtained values.

The Euler-Maclaurin formula provides a controlled approximation of sums with integrals and vice versa \cite{de1970asymptotic}. Suppose $k$ and $\ell$ are natural numbers and $f(x)$ is a real valued continuous function of the number $x\in [ k,\ell]$, then
\begin{eqnarray}
\sum_k^\ell f(x)\approx \int_k^\ell f(x) dx +\frac{f( \ell )-f(k)}{2}+\sum_{h=1}^{\lfloor p/2\rfloor} \frac{B_{2h}}{(2h)!}\left( f^{(2h-1)}(\ell)-f^{(2h-1)}(k)\right)+R, \label{EulerMac}
\end{eqnarray}
where $f^{(2h-1)}$ denotes the $(2h-1)^{\mbox{st}}$ derivative of $f$, and $B_{2h}$ are the Bernoulli numbers. The remainder,  $R$, satisfies
\begin{eqnarray}
|R|\le\frac{2\zeta(p)}{(2\pi)^p}\int_{k}^\ell \left| f^{(p)}(x)\right| dx,\label{R}
\end{eqnarray}
where $\zeta$ is the Riemann zeta function.
\begin{comment}
This theorem will enable us to replace sums with integrals. 
\begin{thm}(Integral Test) Suppose $k$ is an integer, and that the function $f$ is Riemann integrable on every closed, bounded subinterval of $[k,\infty)$. Further, suppose that for $x\ge k $, $f(x)\ge 0$ and is decreasing. Then the series $\sum_{i=k}^\infty f(i)$ converges if and only if the improper integral $\int_k^\infty f(x) dx$ converges. In addition we have the bounds
\begin{eqnarray*}
\int_k^\infty f(x) dx\le \sum_{i=k}^\infty f(i) \le f(k)+\int_k^\infty f(x) dx.
\end{eqnarray*}\label{Thm:IntegralTest}
\end{thm}
\end{comment}

This formula is particularly robust for functions that involve gaussians as a factor. The error term can be zero and often is small as the following lemma shows. We will use the following lemma repeatedly:
\begin{lem}\label{Lem:integralSum}
Let $L\gg 1$, $g>1$ be a fixed positive integer and $a>0$ a real number. We have
\begin{eqnarray*}
\sum_{m=0}^L m^g \exp\left(-\frac{am^2}{L}\right)=\int_0^\infty  m^g \exp\left(-\frac{am^2}{L}\right)\mbox{ } dm + O(L^g \exp{(-aL)}).
\end{eqnarray*}\label{Lem:Asymp}
\end{lem}
\begin{proof}
\begin{eqnarray*}
\sum_{m=0}^L m^g \exp\left(-\frac{am^2}{L}\right)=\sum_{m=0}^\infty m^g \exp\left(-\frac{am^2}{L}\right) - \sum_{m=L+1}^\infty m^g \exp\left(-\frac{am^2}{L}\right)
\end{eqnarray*}
Since the summand is decreasing on $[L,\infty)$, using the integral test, we have
\begin{eqnarray*}
0 \le\sum_{m=L+1}^\infty m^g \exp\left(-\frac{am^2}{L}\right)-\int_{m=L+1}^\infty m^g \exp\left(-\frac{am^2}{L}\right)\mbox{ }dm\le \frac{(L+1)^g \exp\left(-\frac{a(L+1)^2}{L}\right)}{2}.
\end{eqnarray*}
Therefore, 
\begin{eqnarray*}
\sum_{m=0}^L m^g \exp\left(-\frac{am^2}{L}\right)=\sum_{m=0}^\infty m^g \exp\left(-\frac{am^2}{L}\right) +O(L^g \exp{(-aL)}).
\end{eqnarray*}
Since the summand vanishes at zero and infinity, using Euler-Maclauren formula with $p=2$ we have
\begin{eqnarray*}
\sum_{m=0}^\infty m^g \exp\left(-\frac{am^2}{L}\right) = \int_{m=0}^\infty m^g \exp\left(-\frac{am^2}{L}\right) +R,
\end{eqnarray*}
where denoting by $f(m)\equiv m^g \exp\left(-\frac{am^2}{L}\right)$ the error term vanishes because $\frac{2\zeta(2)}{(2\pi)^2}\int_{0}^\infty f^{(2)}(m) dm=0$.
\end{proof}
The rest of this section derives the asymptotic form of the sum in Eqs. \eqref{eq:Mnms} in the large $L$ limit. Later $L$ is replace by the appropriate large parameters $n_1$ or $2n-n_1$. A generalization of the method below is developed in Subsection \ref{sub:Two-point-functions}. The starting point is the summand (with $L\gg 1$) 
\begin{equation}
M_{L,m,i}=\left(m+1\right)\left(\begin{array}{ccc}
 & L\\
i+m+1 & i & L-2i-m
\end{array}\right)\quad.\label{eq:Trinomial_0}
\end{equation}
The saddle point in the $\left(m,i\right)$-plane, must simultaneous satisfy
\[
\begin{array}{ccccccc}
\frac{M_{L,m,i+1}}{M_{L,m,i}} & = & 1, & \qquad & \frac{M_{L,m+1,i}}{M_{L,m,i}} & = & 1\quad.\end{array}
\]
The condition $\frac{M_{L,m,i+1}}{M_{L,m,i}}=1$ gives $\left(L-2i-m\right)^{2}-i\left(i+m\right)\approx0$,
yet $\frac{M_{L,m+1,i}}{M_{L,m,i}}=1$ has its maximum at $m=0$.
Solution of $i$ gives, 
\begin{eqnarray}
i_{sp} & = & \frac{L}{3} -\frac{m}{2}+\frac{m}{8}\left(\frac{m}{L}\right)+\frac{3 m}{128}\left(\frac{m}{L}\right)^{3}+\mathcal{O}\left(L\left(\frac{m}{L}\right)^{5}\right)\label{eq:saddle}\\
 & \approx & \frac{L}{3}-\frac{m}{2}+\frac{m}{8}\left(\frac{m}{L}\right)\mbox{ }.\nonumber 
\end{eqnarray}

Before getting an asymptotic expansion for Eq. \eqref{eq:Trinomial_0},
we consider an example. We will analyze a trinomial coefficient, where
$x+y+z=0$
\begin{eqnarray}
\left(\begin{array}{ccc}
 & L\\
\frac{L}{3}+x\mbox{ } & \mbox{ }\frac{L}{3}+y\mbox{ } & \mbox{ }\frac{L}{3}+z
\end{array}\right) & \approx & 3 ^{L} \sqrt{\frac{2\pi L}{8\pi^{3}\left(\frac{L}{3}+x\right)\left(\frac{L}{3}+y\right)\left(\frac{L}{3}+z\right)}}\label{eq:Trinomial_1}\\
 & \times & \left(\frac{L}{L+3 x}\right)^{\frac{L}{3}+x}\left(\frac{L}{L+3 y}\right)^{\frac{L}{3}+y}\left(\frac{L}{L+3 z}\right)^{\frac{L}{3}+z}\nonumber 
\end{eqnarray}
But,
\begin{eqnarray*}
\left(\frac{L}{L+3 x}\right)^{\frac{L}{3}+x} & = & \exp\left\{ -\left(\frac{L}{3}+x\right)\log\left(1+\frac{3 x}{L}\right)\right\} \\
 & \approx & \exp\left\{ -\left(\frac{L}{3}+x\right)\left(\frac{3x}{L}-\frac{1}{2}\left(\frac{3 x}{L}\right)^{2}\right)\right\} \\
 & \approx & \exp\left\{ -x-\frac{3 x^{2}}{2L}\right\}\\ 
 \left(\frac{L}{L+3 y}\right)^{\frac{L}{3}+y}  & \approx & \exp\left\{ -y-\frac{3 y^{2}}{2L}\right\} \\
 \left(\frac{L}{L+3 z}\right)^{\frac{L}{3}+z}  & \approx & \exp\left\{ -z-\frac{3 z^{2}}{2L}\right\} 
\end{eqnarray*}

In Eq. (\ref{eq:Trinomial_1}), inside the square root is approximately
$\frac{3\sqrt{3}}{2\pi L}$.
Since $x+y+z=0$,
\begin{equation}
\left(\begin{array}{ccc}
 & L\\
\frac{L}{3}+x\mbox{ } & \mbox{ }\frac{L}{3}+y\mbox{ } & \mbox{ }\frac{L}{3}+z
\end{array}\right)\approx\frac{3^{L+1}\sqrt{3}}{2\pi L}\exp\left(-\frac{3}{2}\frac{x^{2}+y^{2}+z^{2}}{L}\right)\label{eq:trinomial_example}
\end{equation}

Now we use this result to evaluate Eq. \eqref{eq:Trinomial_0} by letting
$i+m=\frac{L}{3}+x$ , $i=\frac{L}{3}+y$ and $L-2i-m=\frac{L}{3}+z$.
Since the standard of deviation of multinomial distributions scales
as $\sqrt{L}$, to get a better asymptotic form, we let $i=i_{sp}+\beta\sqrt{L}$
and $m=\alpha\sqrt{L}$. Hence we identify, 
\begin{eqnarray*}
x & = & \left(\beta+\frac{\alpha}{2}\right)\sqrt{L}+\frac{\alpha^{2}}{8}\\
y & = & \left(\beta-\frac{\alpha}{2}\right)\sqrt{L}+\frac{\alpha^{2}}{8}\\
z & = & -2\beta\sqrt{L}-\frac{\alpha^{2}}{4}
\end{eqnarray*}
 Making these substitutions we get $-\frac{3}{2}\frac{x^{2}+y^{2}+z^{2}}{L}=-\frac{3 \alpha^{2}}{4}-9\beta^{2}-\mathcal{O}\left(L^{-1/2}\right)$. Therefore,
using Eq. \eqref{eq:trinomial_example}, Eq. \eqref{eq:Trinomial_0} becomes approximately equal to
\begin{eqnarray*}
M\left(L,\alpha,\beta\right) & \equiv & \frac{ 3^{L+1}\sqrt{3}\mbox{ } \alpha}{2\pi L^{3/2}} \exp\left(-\frac{3 \alpha^{2}}{4}-9\beta^{2}\right).
\end{eqnarray*}

Using the lemma, we replace the sum over $i$ with an integral over with respect to $\sqrt{L} \mbox{ }d\beta$ and perform the resulting gaussian integration around $i_{sp}$ to arrive at the asymptotic form of Eq. \eqref{eq:Mnms}. Substituting $L=n_1$ we have
\begin{eqnarray}
M_{n_{1},0,m}=\frac{m+1}{n_{1}+1}\sum_{i\ge0}\left(\begin{array}{c}
n_{1}+1\\
n_{1}-2i-m\mbox{ }\mbox{ },\mbox{ }i\mbox{ },\mbox{ }\mbox{ }i+m+1
\end{array}\right) & \approx & \frac{3^{n_{1}+3/2}}{2\sqrt{\pi}n_{1}^{3/2}}\mbox{ }\alpha_{1}\exp\left(-\frac{3\alpha_{1}^{2}}{4}\right),\label{eq:MnmsAsym}
\end{eqnarray}
where $\alpha_{1}=m/\sqrt{n_{1}}$. Replacing $n_1$ with $2n-n_1$ an entirely a similar derivation gives
\begin{eqnarray}
M_{2n-n_{1},m,0}=\frac{m+1}{2n-n_{1}+1}\sum_{i\ge0}\left(\begin{array}{c}
2n-n_{1}+1\\
2n-n_{1}-2i-m\mbox{ }\mbox{ },\mbox{ }i\mbox{ },\mbox{ }\mbox{ }i+m+1
\end{array}\right) & \approx & \frac{3^{2n-n_{1}+3/2}}{2\sqrt{\pi}(2n-n_{1})^{3/2}}\mbox{ }\alpha_{2}\exp\left(-\frac{3\alpha_{2}^{2}}{4}\right),\label{eq:MnmsAsym2}
\end{eqnarray}
where $\alpha_{2}=m/\sqrt{(2n-n_{1})}$. 

\begin{rem} In the calculations below the approximation of the sums and evaluation of the resulting integral representations follow the above derivations. More examples and discussion about the approximations of multinomials, as well as, the saddle point technique from an analytic combinatorial perspective can be found in Flajolet and Sedgewick's book \cite{flajolet2009analytic} (See for example Chapter 8.)
\end{rem}
\subsection{Expected height, and $\langle s_{n_{1}}^{x}\rangle$ , $\langle s_{n_{1}}^{y}\rangle$
and $\langle s_{n_{1}}^{z}\rangle$ in the physical limit}
Anisotropy of a Hamiltonian can influence the phase structure \cite{abgaryan2014quantum}. To better understand the anisotropy of the model we prove the following:
\begin{lem}
\label{lem:SxSy}Let $1<n_{1}<2n$  be any site on the chain, then
$\langle s_{n_{1}}^{x}\rangle=\langle s_{n_{1}}^{y}\rangle=0$. \end{lem}
\begin{proof}
Let us first look at $\langle s_{n_{1}}^{x}\rangle=\langle{\cal M}_{2n}|s_{n_{1}}^{x}|{\cal M}_{2n}\rangle$.
The Motzkin state can be written as $|{\cal M}_{2n}\rangle=\frac{1}{\sqrt{N}}\sum_{p=1}^{N}|s_{p}\rangle$,
where $|s_{p}\rangle$ denotes a Motzkin walk. For $\langle{\cal M}_{2n}|s_{n_{1}}^{x}|{\cal M}_{2n}\rangle=\frac{1}{N}\sum_{p}\sum_{p'}\langle s_{p'}|s_{n_{1}}^{x}|s_{p}\rangle$
to be nonzero, it must be that for some $p$ and $p'$, $|s_{p'}\rangle=s^{x}|s_{p}\rangle$;
i.e., any walk $s_{p'}$ and $s_{p}$ must be equal at all $2n-1$
positions excluding the $n_{1}^{st}$ step and the application of
$s^{x}$ at the $n_{1}^{st}$ site should not change the step at that
site. But by Eq. \eqref{eq:Sx}, $s_{n_{1}}^{x}|s_{p}\rangle$ transforms
$|s_{p}\rangle$. Suppose $(s_{p})_{n_{1}}$ is $|+\rangle$ or $|-\rangle$
then $s^{x}|s_{p}\rangle\sim|\tilde{s}_{p}\rangle$ where $(s_{p})_{n_{1}}\ne(\tilde{s}_{p})_{n_{1}}=|0\rangle_{n_{1}}$.
Now suppose $(s_{p})_{n_{1}}=|0\rangle$ and $s_{n_{1}}^{x}|s_{p}\rangle$
gives $s^{x}|0\rangle_{n_{1}}=\frac{1}{\sqrt{2}}\left[|-1\rangle+|+1\rangle\right]$
that is a superposition of two walks one with an excess step up and
one with an excess step down. Both of these walks are not balanced
and cannot be equal to $s_{p'}$. We conclude that $\langle s_{p'}|s_{n_{1}}^{x}|s_{p}\rangle=0$.
An entirely a similar argument applies to $s_{n_{1}}^{y}$ giving
$\langle s_{n_{1}}^{y}\rangle=0$. \end{proof}
\begin{rem}
This proof applies to the generalized model with spin $s>1$ presented
elsewhere \cite{movassagh2016supercritical}.
\end{rem}
Denote by the minimum distance to a boundary by $b\equiv\min\left(2n-n_{1},n_{1}\right)$.  Using Eq. \eqref{eq:Mnms}, we define the probabilities
by 
\[
p_{m}\equiv\frac{M_{n_{1},0,m}M_{2n-n_{1},m,0}}{\sum_{m=0}^{b}M_{n_{1},0,m}M_{2n-n_{1},m,0}}=\frac{m^{2}\exp\left\{ -\frac{3m^{2}}{4}\left[\frac{1}{n_{1}}+\frac{1}{2n-n_{1}}\right]\right\} }{\sum_{m=0}^{b}m^{2}\exp\left\{ -\frac{3m^{2}}{4}\left[\frac{1}{n_{1}}+\frac{1}{2n-n_{1}}\right]\right\} }.
\]
The height expectation value at a distance $b$ from the boundary
is: 
\begin{align}
\langle\widehat{m}_{n_{1}}\rangle & \equiv\sum_{m=0}^{b}m\mbox{ }p_{m}=\frac{\sum_{m=0}^{b}m^{3}\exp\left\{ -\frac{3m^{2}}{4}\left[\frac{1}{n_{1}}+\frac{1}{2n-n_{1}}\right]\right\} }{\sum_{m=0}^{b}m^{2}\exp\left\{ -\frac{3m^{2}}{4}\left[\frac{1}{n_{1}}+\frac{1}{2n-n_{1}}\right]\right\} }.\label{eq:Magnetization_Exact}
\end{align}
The integrals are elementary, with the aid of Lemma \eqref{Lem:integralSum}, they evaluate to give
\begin{eqnarray}
\langle\widehat{m}_{n_{1}}\rangle & \approx & \frac{4}{\sqrt{3\pi}}\sqrt{n_{1}\left(1-\frac{n_{1}}{2n}\right)}\mbox{ }.\label{eq:Magnetization}
\end{eqnarray}
Comment: Clearly, the expected height scales as $\langle\widehat{m}_{n_{1}}\rangle\sim\sqrt{n_{1}}$, which is
expected from the theory of random walks and universality of Brownian motion. See Eq. \eqref{eq:E_m_Excur} in Subsection \ref{sub:BrownianExcursions}
for an alternative  derivation from the theory of Brownian Excursions. 

Application of the binomial expansion to $\langle\Delta\widehat{m}_{n_{1}}\rangle\equiv\langle\widehat{m}_{n_1}-\widehat{m}_{n_{1}-1}\rangle$ gives
\begin{eqnarray}
\langle\Delta\widehat{m}_{n_{1}}\rangle = \frac{2}{\sqrt{3\pi}}\frac{1-n_{1}/n}{\sqrt{n_{1}\left(1-n_{1}/2n\right)}}+\Theta(1/n),\label{eq:Magnetization2}
\end{eqnarray}
where the first term is just the derivative with respect to $n_1$  of $\langle \widehat{m}_{n_{1}}\rangle$ and we have
\begin{equation}
\langle s_{n_{1}}^{z}\rangle\approx\frac{\partial\langle\widehat{m}_{n_{1}}\rangle}{\partial n_{1}}\approx\frac{2}{\sqrt{3\pi}}\frac{1-n_{1}/n}{\sqrt{n_{1}\left(1-n_{1}/2n\right)}}\quad.\label{eq:Sz_mean}
\end{equation}

This shows that the magnetization in this limit (i.e., the bulk), however, vanishes as $n_{1}^{-1/2}$ away from the boundary.

So we have found that the bulk expected magnetization is zero and
that the net small magnetization is propagated from the boundaries
into the bulk. At the boundaries the magnetization is nonzero since
on the left steps down and on the right steps up are forbidden making
the average magnetization positive and negative respectively. See
Fig. \eqref{fig:Sz} for a plot of Eq. \eqref{eq:Sz_mean} and the comparison
of this asymptotic result with the exact sum (Eq. \eqref{eq:Magnetization_Exact})
where the exact $\langle s_{n_{1}}^{z}\rangle\equiv\Delta\langle\widehat{m}_{n_{1}}\rangle=\langle\widehat{m}_{n_{1}}\rangle-\langle\widehat{m}_{n_{1}-1}\rangle$
and the sums of trinomials (Eq. \eqref{eq:Mnms})
were used to obtain $M_{n_{1},0,m}$ and $M_{2n-n_{1},m,0}$ that
appear in Eq. \eqref{eq:Magnetization_Exact}. 
\begin{figure}
\begin{centering}
\includegraphics[scale=0.4]{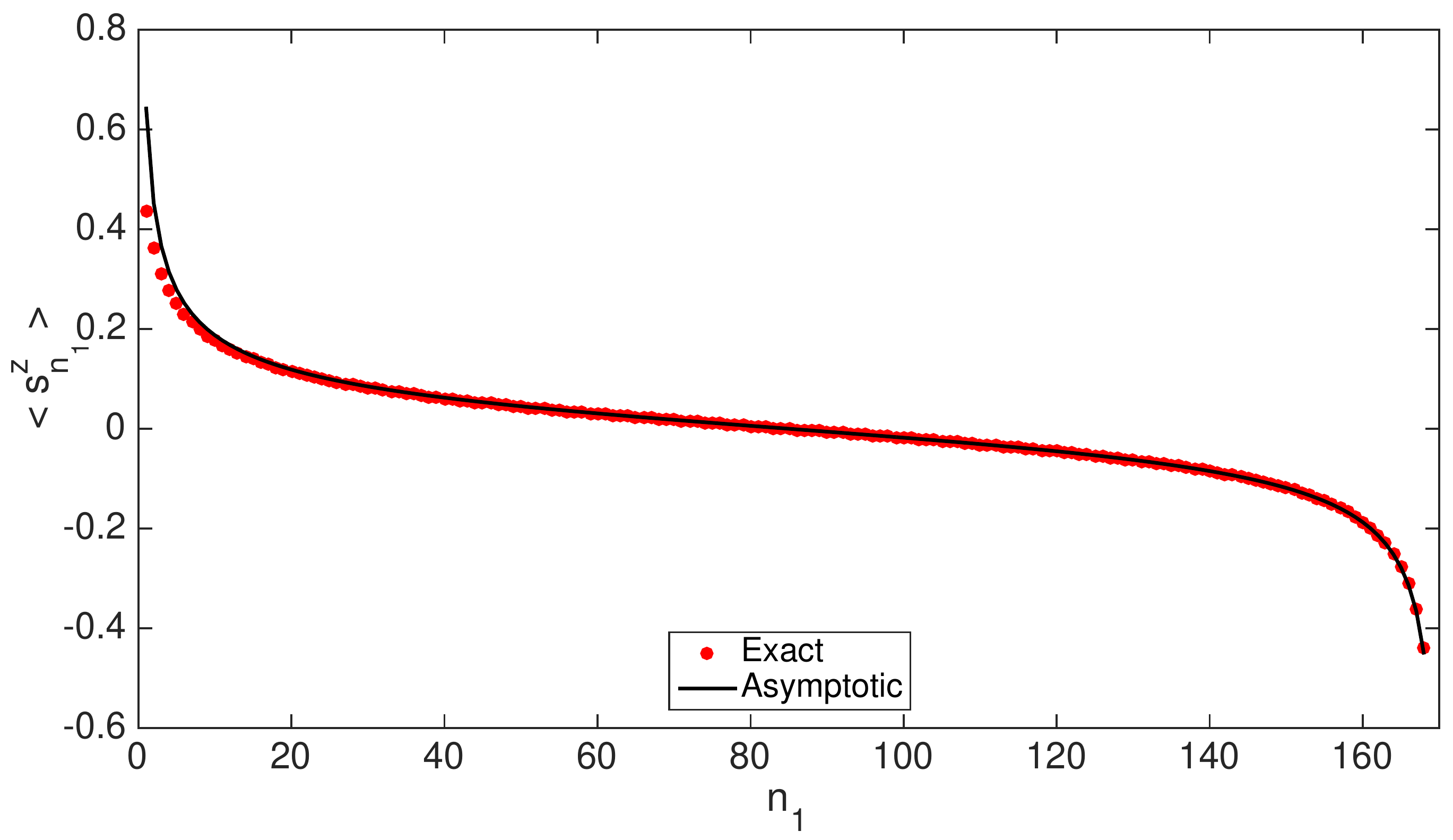}\caption{\label{fig:Sz}$\langle s_{n_{1}}^{z}\rangle$ vs. the location of
the cut, $n_{1}$, on a chain of length $2n=170$. Comparison of the
asymptotic results (Eq. \eqref{eq:Sz_mean}) with the exact expression
for $\Delta\langle m_{n_{1}}\rangle$ (via  Eq. \eqref{eq:Magnetization_Exact}). }
\par\end{centering}
\end{figure}
\subsection{\label{sub:Two-point-functions}Expected height-height and $\langle s_{n-\frac{L}{2}}^{z}s_{n+\frac{L}{2}}^{z}\rangle$
in the physical limit}
We assume that the $L$ consecutive spins are centered on the chain
and denote the distance of the consecutive spins to either boundary
by $b$; i.e., $n_{1}=b$ and $2n-n_{2}=b$, whereby
\begin{eqnarray}
\langle\widehat{m}_{n_{1}}\widehat{m}_{n_{2}}\rangle & = & \frac{\sum_{m_{1}}\mbox{ }\sum_{m_{2}}\mbox{ }m_{1}m_{2}\mbox{ }M_{b,0,m_{1}}M_{L,m_{1},m_{2}}M_{b,m_{2},0}}{\sum_{m_{1}}\mbox{ }\sum_{m_{2}}\mbox{ }M_{b,0,m_{1}}M_{L,m_{1},m_{2}}M_{b,m_{2},0}};\nonumber \\
 & = & \frac{\sum_{m=0}^{b}\mbox{ }\sum_{p=\mbox{max}(-m,-L)}^{\mbox{min}(L,b-m)}\mbox{ }m(m+p)\mbox{ }M_{b,0,m}M_{L,m,(m+p)}M_{b,(m+p),0}}{\sum_{m=0}^{b}\mbox{ }\sum_{p=\mbox{max}(-m,-L)}^{\mbox{min}(L,b-m)}\mbox{ }M_{b,0,m}M_{L,m,(m+p)}M_{b,(m+p),0}};\label{eq:Two_point_physical}
\end{eqnarray}
where the second equality follows from the interdependence of $m_{1}$
and $m_{2}$. Specifically we made the following change of variables
$m_{1}\rightarrow m$ and $m_{2}\rightarrow m+p$ and the limits of
the summation over $p$ look the way they do because in addition to
$m_{1}$ and $m_{2}$ being dependent, the walks on the $L$ consecutive
steps from $n_{1}$ to $n_{2}$ should be non-negative. 

The asymptotic forms of $M_{b,0,m}$ and $M_{b,m+p,0}$ are given
by Eq. \eqref{eq:Mnms}; identifying $m=\alpha_{1}\sqrt{b}$ and $m+p=\alpha_{2}\sqrt{b}$
we have 
\begin{eqnarray}
M_{b,0,m} & \approx & \left\{ \frac{3^{b+3/2}}{2\sqrt{\pi}b^{3/2}}\mbox{ }\alpha_{1}\exp\left(-\frac{3\alpha_{1}^{2}}{4}\right)\right\} =\frac{3^{b+3/2}}{2\sqrt{\pi}b^{2}}\mbox{ }m\mbox{ }\exp\left(-\frac{3}{4}\frac{m^{2}}{b}\right)\mbox{ },\label{eq:Intermediate0}\\
M_{b,m+p,0} & \approx & \left\{ \frac{3^{b+3/2}}{2\sqrt{\pi}b^{3/2}}\mbox{ }\alpha_{2}\exp\left(-\frac{3\alpha_{2}^{2}}{4}\right)\right\} =\frac{3^{b+3/2}}{2\sqrt{\pi}b^{2}}\mbox{ }(m+p)\mbox{ }\exp\left(-\frac{3}{4}\frac{(m+p)^{2}}{b}\right)\mbox{ }.\label{eq:Intermediate1}
\end{eqnarray}

It remains to obtain a good estimate for $M_{L,m,(m+p)}$. The majority
of the walks (probability mass) are centered around a height proportional
to $\sqrt{b}$. Since $L\ll b$, we will not need to subtract 'bad'
walks in Lemma \eqref{lem:GeneralizedBallot} as the following lemma
proves.
\begin{lem}
\label{lem:IgnoreBadWalks}When $1\ll L\ll n$, in Eq. \eqref{eq:Eq_Theorem}
the fraction of 'bad' walks, i.e., walks that become negative on  $L$
steps, is exponentially small in the system's size.\end{lem}
\begin{proof}
For any walk on $L$ steps to be excluded, it must have $m_{1}\le L$,
the number of which is $M_{b,0,m\le L}$, which upper bounds the 'bad'
walks. This number is further upper-bounded by $\sum_{m=0}^{L}\left(\begin{array}{c}
b\\
\frac{b+m}{2}
\end{array}\right)$, which is the total number of walks in Eq. \eqref{eq:Eq_Theorem}.
By Stirling's approximation and for small $m$, we have $\left(\begin{array}{c}
b\\
\frac{b+m}{2}
\end{array}\right)\approx\frac{2^{b+1}}{\sqrt{2\pi b}}\exp\left(-\frac{m^{2}}{2b}\right)$. Therefore $\sum_{m=0}^{L}\left(\begin{array}{c}
b\\
\frac{b+m}{2}
\end{array}\right)\le\frac{2^{b+1}L}{\sqrt{2\pi b}}$. Whereas 
\[
\sum_{m=0}^{b}M_{b,0,m}\approx\frac{3^{b+3/2}}{2\sqrt{\pi}b}\int_{0}^{\infty}d\alpha\mbox{ }\alpha\mbox{ }\exp\left(-\frac{3}{4}\alpha^{2}\right)=\frac{3^{b+1}}{\sqrt{\pi}b}.
\]
We conclude that the ratio of bad walks to all the walks in calculating
$M_{L,m_{1},m_{2}}$ is upper bounded by $L\sqrt{\frac{b}{2}}\left(\frac{2}{3}\right)^{b+1}$,
which is exponentially small in the system's size. Moreover, this
bound is not tight.
\end{proof}
Since in the approximation below, the dependence on $m$ drops out
we have (using $2i=L-k-|p|$)
\begin{eqnarray}
M_{L,p}\equiv M_{L,m,(m+p)} & \approx & \sum_{k=0}^{L-|p|}\left(\begin{array}{c}
L\\
k
\end{array}\right)\left(\begin{array}{c}
L-k\\
\frac{L-k+|p|}{2}
\end{array}\right)=\sum_{i\ge0}\left(\begin{array}{ccccc}
 &  & L\\
L-2i-|p| & , & i+|p| & , & i
\end{array}\right)\equiv\sum_{i\ge0}K_{L,i,p}\label{eq:Stirling}
\end{eqnarray}

The maximum of $K_{L,i,p}$ is at $i=L/3$ and $p=0$. The series
expansion of $i_{sp}$ and $p_{sp}$ are obtained by solving for the
fixed point of $K_{L,i,p}$ in the $ip-$plane. This point is the
simultaneous solution of $\frac{K_{L,i+1,p}}{K_{L,i,p}}=1$ and $\frac{K_{L,i,p+1}}{K_{L,i,p}}=1$.
The solution is $p_{sp}=0$ and $i_{sp}\approx\frac{L}{3}$. Since
the standard of deviation of the multinomials scales as $\sqrt{L}$,
to get better estimates we let $p=\alpha\sqrt{L}$ and $i=i_{sp}+\beta\sqrt{L}$,
which give $i\equiv\frac{L}{3}+\beta\sqrt{L}$ and $p\equiv\alpha\sqrt{L}$.

Let $x+y+z=0$ and $i+p=\frac{L}{3}+x$, $i=\frac{L}{3}+y$ and $L-2i-p=\frac{L}{3}+z$,
we have that 
\[
\left(\begin{array}{ccccc}
 &  & L\\
\frac{L}{3}+z & , & \frac{L}{3}+y & , & \frac{L}{3}+x
\end{array}\right)\approx\frac{3^{L+3/2}}{2\pi L}\exp\left\{ -\frac{3}{2L}\left(x^{2}+y^{2}+z^{2}\right)\right\} ,
\]
where $x=\left(\beta+\alpha\right)\sqrt{L}$, $y=\beta\sqrt{L}$ and
$z=\left(-2\beta-\alpha\right)$. This gives
\[
K_{L,i,p}\approx K(L,\alpha,\beta)=\frac{3^{L+3/2}}{2\pi L}\exp\left\{ -3\alpha^{2}-9\alpha\beta-9\beta^{2}\right\} .
\]
Since the saddle point is away from the boundaries, we integrate this
with respect to $\int di=\sqrt{L}\int_{-\infty}^{\infty}d\beta$ to
get $K(L,\alpha)\approx\frac{3^{L+1/2}}{2\sqrt{\pi L}}\exp\left\{ -\frac{3}{4}\alpha^{2}\right\} .$
This re-expressed in terms of $p$ gives 
\begin{equation}
M_{L,p}\approx\frac{3^{L+1/2}}{2\sqrt{\pi L}}\exp\left[-\frac{3}{4}\frac{p^{2}}{L}\right],\label{eq:Intermediate3}
\end{equation}
which is independent of $m$ as expected. Lemma \eqref{lem:IgnoreBadWalks},
as well as, Eqs. \eqref{eq:Intermediate0}, \eqref{eq:Intermediate1}
and \ref{eq:Intermediate3} are the main results used to derive analytical
formulas for the two-point function and block entanglement entropy
below. Putting these in Eq. \eqref{eq:Two_point_physical} and canceling
constants we get 
\begin{eqnarray}
\langle\widehat{m}_{n_{1}}\widehat{m}_{n_{2}}\rangle & = & \frac{1}{T}\sum_{m=0}^{b}\sum_{p=-L}^{L}m^{2}(m+p)^{2}\exp\left[-\frac{3}{4}\frac{p^{2}}{L}\right]\exp\left[-\frac{3}{4b}\left(m^{2}+\left(m+p\right)^{2}\right)\right],\label{eq:UseFul_Formula}\\
T & \equiv & \sum_{m=0}^{b}\sum_{p=-L}^{L}m(m+p)\exp\left[-\frac{3}{4}\frac{p^{2}}{L}\right]\exp\left[-\frac{3}{4b}\left(m^{2}+\left(m+p\right)^{2}\right)\right].\nonumber 
\end{eqnarray}

Because of the exponential suppression and Lemma \eqref{Lem:integralSum} the limits of
the sums are extended and the sums approximated by integrals. Using the substitution $m\equiv\theta\sqrt{b}$
and $p=\rho\sqrt{L}$ we have
\begin{eqnarray*}
\langle\widehat{m}_{n_{1}}\widehat{m}_{n_{2}}\rangle & \approx & \frac{b}{T'}\int_{-\infty}^{+\infty}d\rho\int_{0}^{\infty}d\theta\mbox{ }\mbox{ }K(\theta,\rho)\mbox{ }\theta\left(\theta+\rho\sqrt{\frac{L}{b}}\right)\\
K(\theta,\rho) & \equiv & \theta b(\theta+\rho\sqrt{\frac{L}{b}})\exp\left[-\frac{3}{4}\rho^{2}\right]\exp\left[-\frac{3}{4}\left(\theta^{2}+\left(\theta+\rho\sqrt{\frac{L}{b}}\right)^{2}\right)\right]\\
T' & \equiv & \int_{-\infty}^{+\infty}d\rho\int_{0}^{\infty}d\theta\mbox{ }K(\theta,\rho)\quad.
\end{eqnarray*}

The integrals are elementary and we integrate over $\theta\in[0,\infty)$
and $\rho\in(-\infty,+\infty)$ to obtain (recall $b=n-L/2$)
\begin{eqnarray}
\langle\widehat{m}_{n-\frac{L}{2}}\widehat{m}_{n+\frac{L}{2}}\rangle & \approx & \frac{b}{1+L/2b}+\frac{2}{3}L=n-\frac{L}{3}+\frac{L^{2}}{4n}\label{eq:2Pt_Final}\\
T & \approx & \frac{4\pi}{9}b^{3}\sqrt{\frac{L}{(2b+L)^{3}}}.\nonumber 
\end{eqnarray}
Comment: From Eq. \eqref{eq:Magnetization}, to the leading order we
have (also compare with Eq. \eqref{eq:Connected_Physical})
\[
\langle\widehat{m}_{n-\frac{L}{2}}\widehat{m}_{n+\frac{L}{2}}\rangle-\langle\widehat{m}_{n-\frac{L}{2}}\rangle\langle\widehat{m}_{n+\frac{L}{2}}\rangle\approx n\left(1-\frac{8}{3\pi}\right).
\]

Recall that $L\equiv n_{2}-n_{1}$, and let $f(n_1,n_2)\equiv \langle\widehat{m}_{n_1}\widehat{m}_{n_2}\rangle$, then the exact value of $\langle s_{n_{1}}^{z}s_{n_{2}}^{z}\rangle$ is given by the bivariate finite difference equation 
\begin{eqnarray}
\langle s_{n_{1}}^{z}s_{n_{2}}^{z}\rangle=\frac{1}{4}\left[f(n_1+1,n_2+1)-f(n_1+1,n_2-1)-f(n_1-1,n_2+1)+f(n_1-1,n_2-1)\right]\label{eq:bivariateTwopoint},
\end{eqnarray}
which  using Eq. \eqref{eq:2Pt_Final} is identically zero. So we have
\begin{eqnarray}
\langle s_{n_{1}}^{z}s_{n_{2}}^{z}\rangle=0
\label{eq:2pt_spinCorr}
\end{eqnarray}

Comment: If we were to approximate $n_1$ and $n_2$ with continuous variables then 
 $\langle s_{n_{1}}^{z}s_{n_{2}}^{z}\rangle\approx\frac{\partial^{2}}{\partial n_{1}\partial n_{n}}\langle\widehat{m}_{n_1}\widehat{m}_{n_2}\rangle$,
which at $n_{1}=n-\frac{L}{2}$ and $n_{2}=n+\frac{L}{2}$ evaluates
to be $ \langle s_{n-\frac{L}{2}}^{z}s_{n+\frac{L}{2}}^{z}\rangle\approx-\frac{1}{2n}$. As with the first derivatives the extension to the continuous variables gives vanishing errors that are $\Theta(1/n)$. Eq. \eqref{eq:2pt_spinCorr}  gives a good agreement against  DMRG calculations. DMRG calculations show that $ \langle s_{n-\frac{L}{2}}^{z}s_{n+\frac{L}{2}}^{z}\rangle$  also vanishes \cite{AdrianChat} in the colored Motzkin model as well \cite{movassagh2016supercritical}.

\subsection{\label{sub:BrownianExcursions}Correlation functions in the Brownian
excursions limit}
In this section we derive the correlation function in the Excursion
limit discussed in the Subsection \ref{sub:Discussion-of-limits}.
The derivations below serve as both an alternative derivation of some
of the formulas derived above,  and derivation of new formulas in
the limit that $L$ tends to infinity simultaneously with $n$. 

In the limit of $n\rightarrow\infty$  the random walk converges to
a Wiener process \cite{prokhorov1956convergence} and a random Motzkin
walk converges to a Brownian excursion \cite{kaigh1976invariance},
denoted by $e(\lambda)$. Mathematically, for $0<\lambda<1$ and $n_{1}\equiv2n\lambda$
\[
\frac{m_{2n\lambda}}{\sqrt{2n\sigma^{2}}}\rightarrow e(\lambda)
\]
where $\sigma^{2}=2/3$ \cite{kaigh1976invariance}. For $\lambda\in(0,1)$,
the probability density of $e(\lambda)$ is \cite{billingsley2013convergence}
\[
f_{\lambda}(x)=2x^{2}\frac{\exp\left[-\frac{x^{2}}{2\lambda(1-\lambda)}\right]}{\sqrt{2\pi\lambda^{3}(1-\lambda)^{3}}}\mathbb{I}_{x\ge0}\quad.
\]
\begin{rem}
Below we denote the expectations with respect to this density by $\mathbb{E}[\centerdot]$,
in contrast to $\langle\centerdot\rangle$, which was used to denote
the expectation with respect to a uniform superposition of all Motzkin
walks. 
\end{rem}
The first two moments of the height are (we do not put hats on $m_{2n\lambda}$
as it is not a quantum operator anymore) 
\begin{eqnarray}
\mathbb{E}[m_{2n\lambda}] & \sim & 2\sqrt{2n\sigma^{2}}\int_{0}^{\infty}dx\mbox{ }x^{3}\frac{\exp\left[-\frac{x^{2}}{2\lambda(1-\lambda)}\right]}{\sqrt{2\pi\lambda^{3}(1-\lambda)^{3}}}=4\sqrt{n}\mbox{ }\sqrt{\frac{2\lambda(1-\lambda)}{3\pi}}\label{eq:E_m_Excur}\\
\mathbb{E}[m_{2n\lambda}^{2}] & \sim & 4n\sigma^{2}\int_{0}^{\infty}dx\mbox{ }x^{4}\frac{\exp\left[-\frac{x^{2}}{2\lambda(1-\lambda)}\right]}{\sqrt{2\pi\lambda^{3}(1-\lambda)^{3}}}=4n(1-\lambda)\lambda\quad.\label{eq:E_m2_Excur}
\end{eqnarray}

Comment: Eq. \eqref{eq:E_m_Excur} coincides with Eq. \eqref{eq:Magnetization}
with the substitution $n_{1}=2\lambda n$. 

Therefore, from the theory of Brownian excursions, we have 
\begin{equation}
\mathbb{E}[m_{2n\lambda}^{2}]-\mathbb{E}^{2}[m_{2n\lambda}]=n\left[4\lambda(1-\lambda)\left(1-\frac{8}{3\pi}\right)\right]\label{eq:connectedCompExcursions}
\end{equation}
Comment: For $\lambda=1/2$, $\mathbb{E}\left[m_{n}^{2}\right]=n$,
which is confirmed to the leading order by our earlier derivations:
\[
\langle\widehat{m}_{n-\frac{L}{2}}\widehat{m}_{n+\frac{L}{2}}\rangle\approx\langle\widehat{m}_{n}^{2}\rangle=n\frac{\int_{0}^{\infty}\alpha^{4}\exp\left(-3\alpha^{2}/2\right)}{\int_{0}^{\infty}\alpha^{2}\exp\left(-3\alpha^{2}/2\right)}=n\quad.
\]

With overwhelming probability the Motzkin walks satisfy $m_{n\pm L/2}\in[m_{n}-C\sqrt{L\mbox{ }\log n}\mbox{ },\mbox{ }m_{n}+C\sqrt{L\mbox{ }\log n}]$
\cite{billingsley2013convergence}. Before we assumed $L$ to be smaller
than all asymptotically increasing functions of $n$. As long as $L=o(n/\log n)$,
the corrections are negligible and we indeed have \cite{billingsley2013convergence,JK_Chat}.
\[
\mathbb{E}[m_{n-\frac{L}{2}}m_{n+\frac{L}{2}}]\approx\mathbb{E}[m_{n}^{2}]=n\quad.
\]

This, yet again, confirms the leading order asymptotic given by Eq.
\ref{eq:2Pt_Final}. Moreover, one does not expect the connected component
of the correlation vanish. Mathematically, we have (either using Eq.
\ref{eq:connectedCompExcursions} with $t=1/2$ or alternatively using
Eqs. \eqref{eq:Magnetization} and \eqref{eq:2Pt_Final}) 
\begin{equation}
\langle\widehat{m}_{n-\frac{L}{2}}\widehat{m}_{n+\frac{L}{2}}\rangle-\langle\widehat{m}_{n-\frac{L}{2}}\rangle\langle\widehat{m}_{n+\frac{L}{2}}\rangle\approx\langle\widehat{m}_{n}^{2}\rangle-\langle\widehat{m}_{n}\rangle^{2}=n\left(1-\frac{8}{3\pi}\right)\quad\mbox{Physical limit}.\label{eq:Connected_Physical}
\end{equation}

In the Excursion limit discussed in Sec. \ref{sub:Discussion-of-limits},
where $n_{1}=2\lambda n$ and $n_{2}=2\mu n$ with $0<\lambda<\mu<1$,
the quantity $L=2n(\mu-\lambda)$ simultaneously tend to infinity
with $n$. In this limit, unlike the physical limit, the number of
``bad'' walks are not negligible. The $2-$point function is still
given by Eq. \eqref{eq:Two_point_physical} and Eq. \eqref{eq:Intermediate0}.
In this limit, and in Eqs. \eqref{eq:Two_point_physical} and \eqref{eq:MotzkinBlockEntanglementEntropy},
the sums are \textit{not} well approximated if we take $p\in[-L,L]$.
We would have to use the limits as in Eq. \eqref{eq:Two_point_physical}. 

We want to calculate $\mathbb{E}[m_{2n\lambda}m_{2n\mu}]-\mathbb{E}[m_{2n\lambda}]\mathbb{E}[m_{2n\mu}]$.
The density for a Brownian excursion to to visit $(\lambda,x_{1})$
and $(\mu,x_{2})$ is \cite{billingsley2013convergence}
\[
f_{\lambda,\mu}(x_{1},x_{2})=2\sqrt{2\pi}p_{0}(\lambda,x_{1})\mbox{ }p(\lambda,x_{1},\mu,x_{2})\mbox{ }p_{0}(1-\mu,x_{2})
\]
where 
\begin{eqnarray*}
p_{0}(\lambda,x_{1}) & = & \frac{x_{1}e^{-\frac{x_{1}^{2}}{2\lambda}}}{\sqrt{2\pi}\lambda^{3/2}}\mathbb{I}_{x_{1}\ge0}\\
p(\lambda,x_{1},\mu,x_{2}) & = & \frac{\exp\left[-\frac{(x_{1}-x_{2})^{2}}{2(\mu-\lambda)}\right]-\exp\left[-\frac{(x_{1}+x_{2})^{2}}{2(\mu-\lambda)}\right]}{\sqrt{2\pi(\mu-\lambda)}}\mathbb{I}_{x_{1}\ge0}\mathbb{I}_{x_{2}\ge0}.
\end{eqnarray*}

Therefore we have for $0<\lambda<\mu<1$ 
\begin{eqnarray*}
f_{\lambda,\mu}(x_{1},x_{2}) & = & 2\sqrt{2\pi}\frac{x_{1}e^{-\frac{x_{1}^{2}}{2\lambda}}}{\sqrt{2\pi}\lambda^{3/2}}\mbox{}\frac{\exp\left[-\frac{(x_{1}-x_{2})^{2}}{2(\mu-\lambda)}\right]-\exp\left[-\frac{(x_{1}+x_{2})^{2}}{2(\mu-\lambda)}\right]}{\sqrt{2\pi(\mu-\lambda)}}\mbox{ }\frac{x_{2}e^{-\frac{x_{2}^{2}}{2(1-\mu)}}}{\sqrt{2\pi}(1-\mu)^{3/2}}\mathbb{I}_{x_{1}\ge0}\mathbb{I}_{x_{2}\ge0}.
\end{eqnarray*}

Comment: One directly verifies that $\int_{0}^{\infty}dx_{1}\int_{0}^{\infty}dx_{2}\mbox{ }f_{\lambda,\mu}(x_{1},x_{2})=1$
as expected. 

We are interested in finding
\[
\mathbb{E}[m_{2n\lambda}m_{2n\mu}]-\mathbb{E}[m_{2n\lambda}]\mathbb{E}[m_{2n\mu}]
\]
Since the Motzkin walk is over $2n$ steps, the expectations ($\sigma^{2}=2/3$)
would be given by 
\begin{eqnarray*}
\mathbb{E}[m_{2n\lambda}m_{2n\mu}] & = & 2n\sigma^{2}\int_{0}^{\infty}dx_{1}\int_{0}^{\infty}dx_{2}\mbox{ }x_{1}x_{2}\mbox{ }f_{\lambda,\mu}(x_{1},x_{2})\\
\mathbb{E}[m_{2n\lambda}]\mathbb{E}[m_{2n\mu}] & = & 2n\sigma^{2}\int_{0}^{\infty}dx\mbox{ }2x^{3}\frac{\exp\left[-\frac{x^{2}}{2\lambda(1-\lambda)}\right]}{\sqrt{2\pi\lambda^{3}(1-\lambda)^{3}}}\int_{0}^{\infty}dx\mbox{ }2x^{3}\frac{\exp\left[-\frac{x^{2}}{2\mu(1-\mu)}\right]}{\sqrt{2\pi\mu^{3}(1-\mu)^{3}}}
\end{eqnarray*}

Direct computation of these gives
\begin{eqnarray*}
\mathbb{E}[m_{2n\lambda}m_{2n\mu}] & = & \frac{4n\sigma^{2}}{\pi}\left\{ 3\sqrt{\lambda(1-\mu)(\mu-\lambda)}+\left[\lambda(2-3\mu)+\mu\right]\arctan\left(\sqrt{\frac{\lambda(1-\mu)}{\mu-\lambda}}\right)\right\} \\
\mathbb{E}[m_{2n\lambda}]\mathbb{E}[m_{2n\mu}] & = & \frac{16n\sigma^{2}}{\pi}\sqrt{\lambda\mu(1-\mu)(1-\lambda)}
\end{eqnarray*}
We see that the disconnected components do \textit{not} cancel, 
\begin{eqnarray*}
\mathbb{E}[m_{2n\lambda}m_{2n\mu}]-\mathbb{E}[m_{2n\lambda}]\mathbb{E}[m_{2n\mu}] & = & \frac{4n\sigma^{2}}{\pi}\left\{ 3\sqrt{\lambda(1-\mu)(\mu-\lambda)}+\left[\lambda(2-3\mu)+\mu\right]\arctan\left(\sqrt{\frac{\lambda(1-\mu)}{\mu-\lambda}}\right)\right.\\
 &  & \left.-4\sqrt{\lambda\mu(1-\mu)(1-\lambda)}\right\} \qquad\mbox{Excursion limit}.
\end{eqnarray*}

Now restoring back $n_{1}=2n\lambda$ and $n_{2}=2n\mu$ we find that
even in this limit 
\begin{equation}
\langle s_{n_{1}}^{z}s_{n_{2}}^{z}\rangle=\frac{\partial^{2}}{\partial n_{1}\partial n_{2}}\mathbb{E}[m_{n_{1}}m_{n_{2}}]=\mathcal{O}(n^{-1})\qquad\mbox{Excursion limit}.\label{eq:SzSz_Excursion}
\end{equation}

Comment: The connected component also follows the same asymptotic
scaling, i.e., $\frac{\partial^{2}}{\partial n_{1}\partial n_{2}}\left\{ \mathbb{E}[m_{n_{1}}m_{n_{2}}]-\mathbb{E}[m_{n_{1}}]\mathbb{E}[m_{n_{2}}]\right\} =\mathcal{O}(n^{-1})$.
\section{Entanglement entropies and Schmidt ranks}
In \cite{Movassagh2012_brackets}, it was shown that the half-chain
von Neumann entanglement entropy and Schmidt rank are 
\begin{eqnarray*}
S_{n} & = & \frac{1}{2}\log_{2}n+\left(\gamma-\frac{1}{2}\right)\log_{2}e+\frac{1}{2}\log_{2}\left(\frac{2\pi}{3}\right)\quad\mbox{bits},\\
\chi & = & n+1.
\end{eqnarray*}
Below we calculate the bipartite entanglement entropy about any \textit{cut,}
$1\ll n_{1}\ll2n$. 
\subsection{Bipartite entanglement about an arbitrary cut}
Suppose we cut the chain into two parts $A$ and $B$, where $A$
consists of the first $n_{1}$ spins and $B$ the remaining $2n-n_{1}$.
We first show that the Schmidt decomposition of the ground state is
\begin{eqnarray}
|{\cal M}_{2n}\rangle & = & \sum_{m=0}^{b}\sqrt{p_{m}}\mbox{ }|C_{n_{1},0,m}\rangle\otimes|C_{2n-n_{1},m,0}\rangle\label{eq:Bipartite_Motzkin_raw}\\
p_{m} & = & \frac{M_{n_{1},0,m}M_{2n-n_{1},m,0}}{\sum_{m=0}^{n}M_{n_{1},0,m}M_{2n-n_{1},m,0}},\nonumber 
\end{eqnarray}
where $C_{\ell,p,q}$ is a \textit{normalized} uniform superposition
of non-negative ``Motzkin'' walks on $\ell$ steps starting at height
$p$ and ending at height $q$. 

We can organize the Motzkin walks based on the height they have at
the site $n_{1}$ and denoted that height by $m$. So Eq. \eqref{eq:GS}
is equivalent to (recall that $b=\min(n_{1},2n-n_{1})$) 
\[
|{\cal M}_{2n}\rangle=\frac{1}{\sqrt{N}}\sum_{m=0}^{b}\sum_{s_{m}\in\mbox{Motzkin}}|s_{m}\rangle
\]
 where $s_{m}$ is a Motzkin walk that attains the height $m$ at
site $n_{1}$. From this expression we see that the Schmidt rank denoted
by $\chi_{n_{1}}$ is 
\begin{equation}
\chi_{n_{1}}=b+1\mbox{ }.\label{eq:chi_n1}
\end{equation}

For any given $m$, the sum $\sum_{s_{m}}|s_{m}\rangle=\sum_{x=1}^{M_{n_{1},0,m}}|w_{n_{1},0,m}^{x}\rangle\otimes\sum_{y=1}^{M_{2n-n_{1},0,m}}|w_{2n-n_{1},m,0}^{y}\rangle$,
where $\sum_{x=1}^{M_{\ell,u,v}}|w_{\ell,u,v}^{x}\rangle$ is the
\textit{un}normalized sum over all non-negative walks on $\ell$ steps
starting from height $u$ and ending at height $v$. We have 
\begin{equation}
|{\cal M}_{2n}\rangle=\frac{1}{\sqrt{N}}\sum_{m=0}^{b}\left(\sum_{x=1}^{M_{n_{1},0,m}}|w_{n_{1},0,m}^{x}\rangle\otimes\sum_{y=1}^{M_{2n-n_{1},0,m}}|w_{2n-n_{1},m,0}^{y}\rangle\right)\label{eq:Bipartite_Motzkin}
\end{equation}

The reduced density matrix about the cut made at $n_{1}$ is 
\[
\rho_{\mbox{cut}}\equiv\mbox{Tr}_{(1\cdots n_{1})}\left[\rho\right]=\frac{1}{N}\sum_{m=0}^{b}\sum_{m'=0}^{b}\left(\sum_{x=1}^{M_{n_{1},0,m}}\sum_{u=1}^{M_{n_{1},0,m'}}\langle w_{n_{1},0,m'}^{u}|w_{n_{1},0,m}^{x}\rangle\right)\otimes\left(\sum_{y=1}^{M_{2n-n_{1},0,m}}|w_{2n-n_{1},m,0}^{y}\rangle\otimes\sum_{v=1}^{M_{2n-n_{1},0,m'}}\langle w_{2n-n_{1},m',0}^{v}|\right)
\]
But any two distinct walks must be orthogonal so 
\[
\sum_{x=1}^{M_{n_{1},0,m}}\sum_{u=1}^{M_{n_{1},0,m'}}\langle w_{n_{1},0,m'}^{u}|w_{n_{1},0,m}^{x}\rangle=\sum_{x=1}^{M_{n_{1},0,m}}\sum_{u=1}^{M_{n_{1},0,m'}}\delta_{w_{n,0,m'}^{u},w_{n,0,m}^{x}}\delta_{m,m'}=M_{n_{1},0,m}\delta_{m,m'}.
\]
We infer that 
\begin{eqnarray*}
\rho_{\mbox{cut}} & = & \frac{1}{N}\sum_{m=0}^{b}\sum_{m'=0}^{b}\left(M_{n_{1},0,m}\delta_{m,m'}\right)\otimes\left(\sum_{y=1}^{M_{2n-n_{1},0,m}}|w_{2n-n_{1},m,0}^{y}\rangle\otimes\sum_{v=1}^{M_{2n-n_{1},0,m'}}\langle w_{2n-n_{1},m',0}^{v}|\right)\\
 & = & \frac{1}{N}\sum_{m=0}^{n}M_{n,0,m}\left(\sum_{y=1}^{M_{2n-n_{1},0,m}}|w_{2n-n_{1},m,0}^{y}\rangle\otimes\sum_{v=1}^{M_{2n-n_{1},0,m}}\langle w_{2n-n_{1},m,0}^{v}|\right)
\end{eqnarray*}
Since $\sum_{x=1}^{M_{2n-n_{1},0,m}}|w_{2n-n_{1},m,0}^{y}\rangle=\sum_{v=1}^{M_{2n-n_{1},0,m}}|w_{2n-n_{1},m,0}^{v}\rangle=\sqrt{M_{2n-n_{1},m,0}}|C_{2n-n_{1},m,0}\rangle$,
we now have 
\[
\rho_{\mbox{cut}}=\frac{1}{N}\sum_{m=0}^{b}M_{n_{1},0,m}M_{2n-n_{1},m,0}|C_{2n-n_{1},m,0}\rangle\langle C_{2n-n_{1},m,0}|.
\]

This is the desired result where the reduced density matrix is diagonal
in the basis $|C_{2n-n_{1},m,0}\rangle$. The von Neumann entanglement
entropy is 
\begin{eqnarray}
S_{\mbox{cut}} & = & -\mbox{Tr}\left[\rho_{\mbox{cut}}\log\rho_{\mbox{cut}}\right]=-\sum_{m=0}^{b}p_{m}\log p_{m},\label{eq:Scut_Exact}\\
p_{m} & = & \frac{M_{n_{1},0,m}M_{2n-n_{1},m,0}}{N}\quad.\nonumber 
\end{eqnarray}

Using the asymptotic expressions given by Eq. \eqref{eq:Mnms} and canceling constants $p_{m}$ reads
\begin{align}
p_{m}= & \frac{m^{2}\exp\left\{ -\frac{3m^{2}}{4}\left[\frac{1}{n_{1}}+\frac{1}{2n-n_{1}}\right]\right\} }{\sum_{m=0}^{b}m^{2}\exp\left\{ -\frac{3m^{2}}{4}\left[\frac{1}{n_{1}}+\frac{1}{2n-n_{1}}\right]\right\} }.\label{eq:Magnetization_Exact-1}
\end{align}
Because of Lemma \eqref{Lem:integralSum} we can approximate the sums with integrals and extend the upper limit.
To get better estimates, let $m=\alpha\sqrt{n_{1}}$, and
\begin{eqnarray*}
S_{\mbox{cut}} & \approx & -\frac{1}{T'}\int_{0}^{\infty}d\alpha\mbox{ }\alpha^{2}\exp\left[-\frac{3\alpha^{2}}{4}\left(\frac{2n}{2n-n_{1}}\right)\right]\log\left\{ \frac{1}{\sqrt{n_{1}}\mbox{ }T'}\alpha^{2}\exp\left[-\frac{3\alpha^{2}}{4}\left(\frac{2n}{2n-n_{1}}\right)\right]\right\} \\
T' & \approx & \int_{0}^{\infty}d\alpha\mbox{ }\alpha^{2}\exp\left\{ -\frac{3\alpha^{2}}{4}\left[\frac{2n}{2n-n_{1}}\right]\right\} 
\end{eqnarray*}
These integrals are evaluated to give the desired result 
\begin{figure}
\begin{centering}
\includegraphics[scale=0.4]{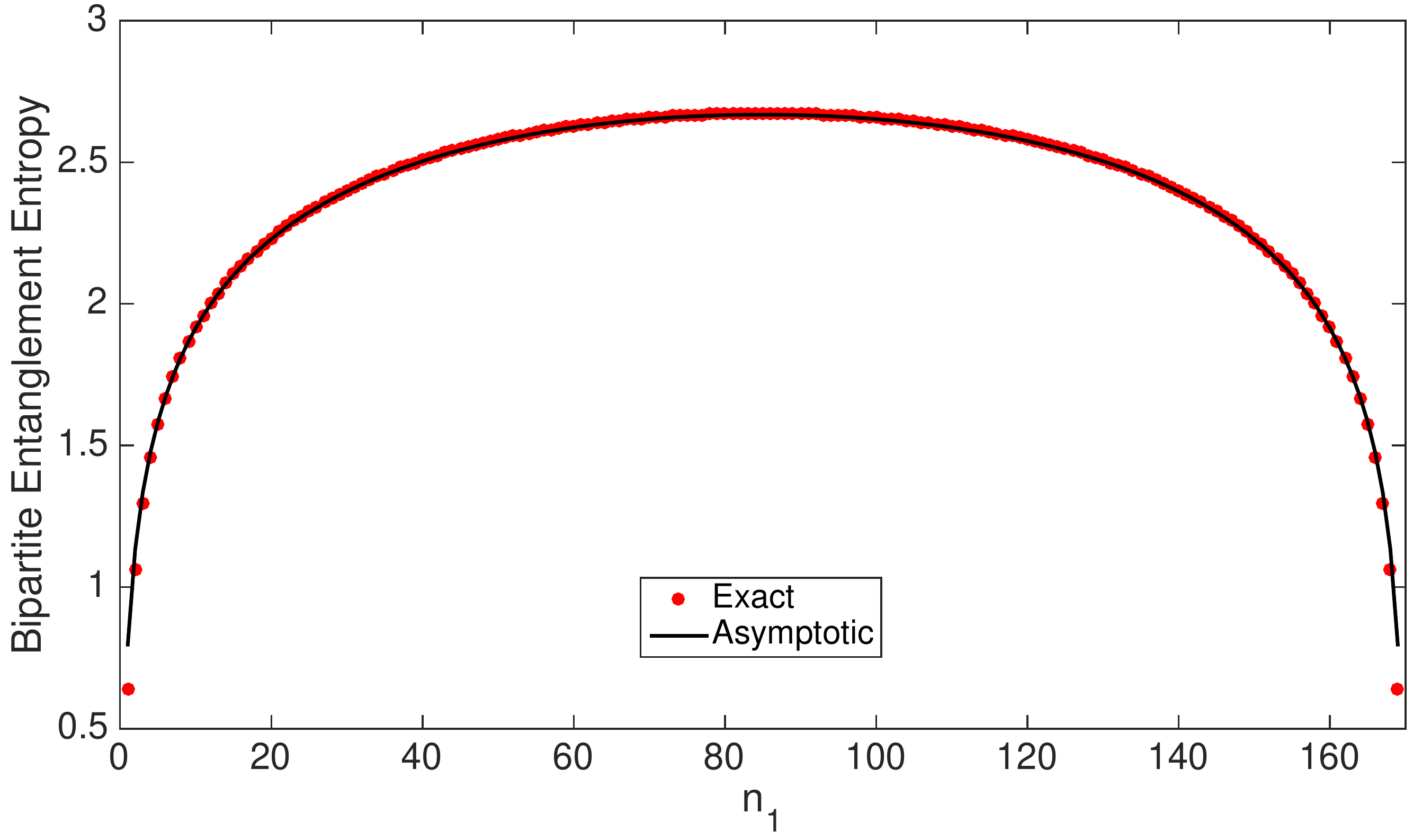}
\par\end{centering}
\caption{\label{fig:Bipartite-entanglement-entropy,}Bipartite entanglement
entropy, $S_{\mbox{cut}}$ vs. the location of the cut, $n_{1}$,
on a chain of length $2n=170$. Comparison of the asymptotic results
(Eq. \eqref{eq:Scut}) with the exact sum given by Eq. \eqref{eq:Scut_Exact}.}
\end{figure}
\begin{eqnarray}
S_{\mbox{cut}} & \approx & \frac{1}{2}\log\left[\frac{n_{1}\left(2n-n_{1}\right)}{n}\right]+\gamma-\frac{1}{2}+\frac{1}{2}\log\left(\frac{2\pi}{3}\right)\qquad\mbox{nats}\label{eq:Scut}\\
 & = & \frac{1}{2}\log_{2}\left[\frac{n_{1}\left(2n-n_{1}\right)}{n}\right]+\left(\gamma-\frac{1}{2}\right)\log_{2}e+\frac{1}{2}\log_{2}\left(\frac{2\pi}{3}\right)\quad\mbox{bits}\nonumber 
\end{eqnarray}
where $\gamma$ is  Euler's constant. 

Note that if we set $n_{1}=n$, we recover the half-chain entanglement
entropy formula in \cite{Movassagh2012_brackets}. As an illustration
in Fig. \eqref{fig:Bipartite-entanglement-entropy,}, we plot Eq.
\ref{eq:Scut} vs. $n_{1}$, for a chain of length $170$ and compare
it with the exact entropy (Eq. \eqref{eq:Scut_Exact}), where the sum over trinomials (Eq. \eqref{eq:Mnms}) were used
to obtain $M_{n_{1},0,m}$ and $M_{2n-n_{1},m,0}$. 

We now calculate the R\'enyi entropy defined by ($\kappa>0$)
\begin{equation}
S_{\mbox{cut}}^{\kappa}\equiv\frac{1}{1-\kappa}\log\left[\mbox{Tr}\left(\rho_{\mbox{cut}}^{\kappa}\right)\right].\label{eq:RenyiCut}
\end{equation}
A very similar calculation as above gives 
\begin{eqnarray}
S_{\mbox{cut}}^{\kappa}=\frac{1}{1-\kappa}\log\sum_{m=0}^{b}p_{m}^{\kappa} & \approx & \frac{1}{2}\log\left[\frac{n_{1}\left(2n-n_{1}\right)}{n}\right]+f(\kappa)\quad\mbox{nats}\label{eq:Scut_Renyi}\\
 & = & \frac{1}{2}\log_{2}\left[\frac{n_{1}\left(2n-n_{1}\right)}{n}\right]+f(\kappa)\log_{2}e\quad\mbox{bits},\\
f(\kappa) & \equiv & \frac{\log\left[\Gamma\left(\kappa+1/2\right)\right]}{1-\kappa}-\frac{1}{2\left(1-\kappa\right)}\left\{ (1+2\kappa)\log\kappa+\kappa\log(\frac{\pi}{24})+\log6\right\} .\nonumber 
\end{eqnarray}

Comment: Indeed $\lim_{\kappa\rightarrow1}S_{\mbox{cut}}^{\kappa}=S_{\mbox{cut}}$
recovers the von Neumann entropy as expected. 
\begin{rem}
Since we have analytically diagonalized the reduced density matrix,
we can identify the Entanglement Hamiltonian defined by \cite{li2008entanglement}
$\rho_{\mbox{cut}}=\exp\left(-\frac{H_{\mbox{cut}}}{T}\right)$. $T$
here denotes the temperature and we have
\begin{eqnarray*}
H_{\mbox{cut}} & = & -T\mbox{ }\log\rho_{\mbox{cut}}=T\log N-T\sum_{m=0}^{b}\log\left[M_{n_{1},0,m}M_{2n-n_{1},m,0}\right]|C_{2n-n_{1},m,0}\rangle\langle C_{2n-n_{1},m,0}|
\end{eqnarray*}
From Eq. \eqref{eq:Magnetization_Exact-1} (subtracting an overall  constant)
we have 
\begin{eqnarray}
H_{\mbox{cut}} & \approx & T\sum_{m=0}^{b}E_{m}|C_{2n-n_{1},m,0}\rangle\langle C_{2n-n_{1},m,0}|\label{eq:Hcut}\\
E_{m} & \equiv & \frac{3m^{2}}{4}\left[\frac{1}{n_{1}}+\frac{1}{2n-n_{1}}\right]-2\mbox{ }\log m.\nonumber 
\end{eqnarray}
It would be interesting if a local Hamiltonian can be identified that
has $E_{m}$ as its spectrum. 
\end{rem}
\subsection{\label{sub:Block-Entanglement-Entropy}Block entanglement}
We now turn to the entanglement entropy of $L-$consecutive spins
centered in the middle of the chain. Let the first $b$ spins be the
subsystem $A$, the next $L$ spins the subsystem $B$ and the remaining
$b$ spins subsystem $C$, i.e., $2n=2b+L$. Since most of the mass
in the summation at the boundaries of $B$ is concentrated around
$m\propto\sqrt{b}$, and $L\ll b$, we re-express Eq. \eqref{eq:GS}
equivalently as
\begin{equation}
|{\cal M}_{2n}\rangle=\frac{1}{\sqrt{N}}\sum_{m=0}^{b}\sum_{p=-L}^{+L}\left(\sum_{x=1}^{M_{b,0,m}}|w_{b,0,m}^{x}\rangle_{A}\otimes\sum_{y=1}^{M_{L,m,m+p}}|w_{L,m,m+p}^{y}\rangle_{B}\otimes\sum_{z=1}^{M_{b,0,m+p}}|w_{b,m+p,0}^{z}\rangle_{C}\right)\label{eq:MotzkinBlockEntanglementEntropy}
\end{equation}
where as before $\sum_{x=1}^{M_{\ell,u,v}}|w_{\ell,u,v}^{x}\rangle$
is the \textit{un}normalized sum over all non-negative walks on $\ell$
steps starting from height $u$ and ending at height $v$. Because
of Lemma \eqref{lem:IgnoreBadWalks} we write 
\begin{equation}
|{\cal M}_{2n}\rangle\approx\frac{1}{\sqrt{N}}\sum_{m=0}^{b}\sum_{p=-L}^{+L}\left(\sum_{x=1}^{M_{b,0,m}}|w_{b,0,m}^{x}\rangle_{A}\otimes\sum_{y=1}^{M_{L,p}}|w_{L,p}^{y}\rangle_{B}\otimes\sum_{z=1}^{M_{b,0,m+p}}|w_{b,m+p,0}^{z}\rangle_{C}\right)\label{eq:MotzkinBlockEntanglementEntropyApprox}
\end{equation}

The reduced density matrix of the $L-$middle spins is $\rho_{B}=\mbox{Tr}_{A,C}\left[\rho\right]$
\begin{eqnarray*}
\rho_{B} & = & \frac{1}{N}\sum_{m,m'=0}^{b}\sum_{p,p'=-L}^{+L}\left\{ \left(\sum_{x=1}^{M_{b,0,m}}\sum_{u=1}^{M_{b,0,m'}}\mbox{}_{A}\langle w_{b,0,m'}^{u}|w_{b,0,m}^{x}\rangle_{A}\right)\left(\sum_{v=1}^{M_{b,0,m'+p'}}\sum_{z=1}^{M_{b,0,m+p}}\mbox{}_{C}\langle w_{b,m'+p',0}^{v}|w_{b,m+p,0}^{z}\rangle_{C}\right)\right.\\
 &  & \quad\left.\left(\sum_{y=1}^{_{M_{L,p}}}|w_{L,p}^{y}\rangle_{B}\otimes\sum_{k=1}^{M_{L,p'}}\mbox{}_{B}\langle w_{L,p'}^{k}|\right)\right\} 
\end{eqnarray*}
But 
\begin{eqnarray*}
\sum_{x=1}^{M_{b,0,m}}\sum_{u=1}^{M_{b,0,m'}}\mbox{}_{A}\langle w_{b,0,m'}^{u}|w_{b,0,m}^{x}\rangle_{A} & = & \sum_{x=1}^{M_{b,0,m}}\sum_{u=1}^{M_{b,0,m'}}\delta_{u,x}\delta_{m,m'}=M_{b,0,m}\delta_{m,m'},
\end{eqnarray*}
and similarly
\begin{eqnarray*}
\sum_{z=1}^{M_{b,0,m+p}}\sum_{v=1}^{M_{b,0,m'+p'}}\mbox{}_{C}\langle w_{b,m'+p',0}^{v}|w_{b,m+p,0}^{z}\rangle_{C} & = & M_{b,0,m+p}\delta_{m+p,m'+p'}\quad.
\end{eqnarray*}
Hence we have 
\begin{eqnarray*}
\rho_{B} & = & \frac{1}{N}\sum_{m,m'=0}^{b}\sum_{p,p'=-L}^{L}M_{b,0,m}M_{b,m+p,0}\delta_{m,m'}\delta_{m+p,m'+p'}\mbox{ }\sum_{y=1}^{_{M_{L,p}}}|w_{L,p}^{y}\rangle_{B}\otimes\sum_{k=1}^{M_{L,p'}}\mbox{}_{B}\langle w_{L,p'}^{k}|\\
 & = & \frac{1}{N}\sum_{p=-L}^{L}M_{L,p}\sum_{m=0}^{b}M_{b,0,m}M_{b,m+p,0}\mbox{ }|C_{L,p}\rangle_{B}\langle C_{L,p}|
\end{eqnarray*}
because $\sum_{y=1}^{M_{L,m,m+p}}|w_{L,p}^{y}\rangle=\sum_{k=1}^{M_{L,p}}|w_{L,p}^{k}\rangle=\sqrt{M_{L,p}}\mbox{ }|C_{L,p}\rangle$.
We derived the asymptotic forms of these in Eqs. \eqref{eq:Intermediate0},
\eqref{eq:Intermediate1} and \ref{eq:Intermediate3}. Using what ultimately
lead to Eq. \eqref{eq:UseFul_Formula} we have 
\begin{eqnarray*}
\rho_{B} & = & \frac{1}{T}\sum_{p=-L}^{L}\mbox{ }\sum_{m=0}^{b}K_{m,p}|C_{L,p}\rangle\langle C_{L,p}|\\
K_{m,p} & \equiv & m(m+p)\exp\left[-\frac{3}{4}\frac{p^{2}}{L}\right]\exp\left[-\frac{3}{4b}\left(m^{2}+\left(m+p\right)^{2}\right)\right]\\
T & \equiv & \sum_{p=-L}^{L}\mbox{ }\sum_{m=0}^{b}K_{m,p}\quad.
\end{eqnarray*}
Because $p\ll m$, $K_{m,p}\approx m^{2}\exp\left[-\frac{3}{4}\frac{p^{2}}{L}\right]\exp\left[-\frac{3m^{2}}{2b}\right]$
and we have\footnote{Without making this approximation in Eq. \eqref{eq:Entropy_Final} we
would find $S_{L}\approx\frac{1}{2}\log L+\log\left(2\sqrt{\frac{\pi}{3}}\right)+\frac{1}{2}-\frac{3}{4}\frac{L}{b}-\frac{9}{16}\left(\frac{L}{b}\right)^{2}+\mathcal{O}\left(\frac{L}{b}\right)^{3}$.}
\begin{eqnarray}
\rho_{B} & \approx & \frac{\sum_{p=-L}^{L}\exp\left(-\frac{3}{4}\frac{p^{2}}{L}\right)\mbox{ }|C_{L,p}\rangle\langle C_{L,p}|}{\sum_{p=-L}^{L}\exp\left(-\frac{3}{4}\frac{p^{2}}{L}\right)}\quad.\label{eq:rhoB}
\end{eqnarray}
We have diagonalized the reduced density matrix in $|C_{L,p}\rangle$
basis and the eigenvalues (i.e., Schmidt numbers) are 
\begin{equation}
\lambda_{p}=\frac{\exp\left(-\frac{3}{4}\frac{p^{2}}{L}\right)}{\sum_{p=-L}^{L}\exp\left(-\frac{3}{4}\frac{p^{2}}{L}\right)}\quad.\label{eq:Lambda}
\end{equation}

Before obtaining asymptotic for the entanglement entropy, from Eq.
\ref{eq:rhoB} we find that the Schmidt rank of the $L$ consecutive
spins, denoted by $\chi_{L}$, is 
\begin{equation}
\chi_{L}=2L+1\mbox{ }.\label{eq:chi_L}
\end{equation}

To make the integrals $\mathcal{O}(1)$, we let $p=\sqrt{L}\rho$
and approximate the sums with integrals over $dp=\sqrt{L}\mbox{ }d\rho$.
Since the maxima is at zero and $L\gg1$, we can extend the limits
of the integral. The von Neumann entanglement entropy of the $L-$consecutive
middle spins in the Motzkin state is 
\begin{eqnarray}
S_{L}\equiv-\mbox{Tr}\left(\rho_{B}\log\rho_{B}\right) & \approx & -\frac{\int_{-\infty}^{+\infty}d\rho\mbox{ }\exp\left(-\frac{3}{4}\rho^{2}\right)\mbox{ }\log\left[\frac{\exp\left(-\frac{3}{4}\rho^{2}\right)}{\sqrt{L}\int_{-\infty}^{+\infty}d\rho\exp\left(-\frac{3}{4}\rho^{2}\right)}\right]}{\int_{-\infty}^{+\infty}d\rho\exp\left(-\frac{3}{4}\rho^{2}\right)}\quad.\label{eq:S_definition}
\end{eqnarray}
$\sqrt{L}$ inside the log already gives the logarithmic scaling of
$S_{L}$ with $L$ and $\int_{-\infty}^{+\infty}d\rho\exp\left(-\frac{3}{4}\rho^{2}\right)=2\sqrt{\frac{\pi}{3}}$.

Comment: The summands are even function and simple application of the integral test shows that the summation is well approximated by the integrals.

So we have in the limit that  $1\ll L\ll n$
\begin{eqnarray}
S_{L} & \approx & \frac{1}{2}\log L+\log\left(2\sqrt{\frac{\pi}{3}}\right)+\frac{1}{2}\quad\mbox{nats}.\label{eq:Entropy_Final}\\
 & = & \frac{1}{2}\log_{2}L+\log_{2}\left(2\sqrt{\frac{\pi}{3}}\right)+\frac{1}{2}\log_{2}e\quad\mbox{bits .}\nonumber 
\end{eqnarray}

Comment: This formula gives good agreements with DMRG calculations
to be presented elsewhere \cite{AdrianChat}.

The formula derived above (Eq. \eqref{eq:Entropy_Final}) is derived
for $L$ consecutive spins centered on a chain of length $2n\gg L$.
However, we believe the same scaling would hold in general: 
\begin{conjecture}
The entanglement entropy of any $L\ll2n$ consecutive block of spins
(not necessarily in the middle), to the leading order, scales as $\log L$. 
\end{conjecture}
Next we calculate the R\'enyi entropy ($\kappa>0$) 
\begin{equation}
S_{L}^{\kappa}\equiv\frac{1}{1-\kappa}\log\left[\mbox{Tr}\left(\rho_{L}^{\kappa}\right)\right].\label{eq:Reyni_Block}
\end{equation}
A very similar calculation as above gives 
\begin{eqnarray}
S_{L}^{\kappa}=\frac{1}{1-\kappa}\log\sum_{p=-L}^{+L}\lambda_{p}^{\kappa} & \approx & \frac{1}{2}\log(L)+\log\left(2\sqrt{\frac{\pi}{3}}\right)-\frac{\log(\kappa)}{2(1-\kappa)}\quad\mbox{nats}\label{eq:SL_Renyi}\\
 & = & \frac{1}{2}\log_{2}(L)+\log_{2}\left(2\sqrt{\frac{\pi}{3}}\right)-\frac{\log_{2}(\kappa)}{2(1-\kappa)}\quad\mbox{bits}.\nonumber 
\end{eqnarray}

Comments: One can verify that $\lim_{\kappa\rightarrow1}S_{L}^{\kappa}=S_{L}$
as expected. Inside the Table in Sec. \ref{sec:Context-and-summary}
we defined 
\[
g(\kappa)\equiv\log\left(2\sqrt{\frac{\pi}{3}}\right)-\frac{\log(\kappa)}{2(1-\kappa)}.
\]
\begin{rem}
The R\'enyi entropies (Eqs. \eqref{eq:Scut_Renyi} and \eqref{eq:SL_Renyi})
depend on $\kappa$ only in the correction terms. This feature is
shared by the AKLT model as well \cite{fan2004entanglement,korepin2010entanglement}.
However, the logarithmic divergence as $\kappa\rightarrow0$ is a
new feature of this model.
\end{rem}
\begin{rem}
Since we have diagonalized the reduced density matrix, we can identify
the Entanglement Hamiltonian defined by \cite{li2008entanglement}
$\rho_{B}\equiv\exp\left(-\frac{H_{L}}{T}\right)$ and we have using
Eq. \eqref{eq:Lambda} and subtracting the overall constant 
\begin{eqnarray}
H_{L} & = & T\sum_{p=-L}^{L}E_{p}\mbox{ }|C_{L,p}\rangle\langle C_{L,p}|\quad.\label{eq:H_L}\\
E_{p} & \equiv & \frac{3}{4}\frac{p^{2}}{L}\quad.\nonumber 
\end{eqnarray}
\end{rem}

\section{Conclusions and future work}
In recent years the interplay between condensed matter physics and
quantum information theory has been quite fruitful. The model proposed
in \cite{Movassagh2012_brackets} is a new exactly solvable model
in physics that owes much of its novelty to ideas and techniques of
quantum information theory as well as other areas of mathematics and
computer science. For example, it utilizes the theory of Brownian
motions and random walks, fractional matching technique in computer
science, perturbation theory and asymptotic analysis. The model is
exactly solvable in the sense that the exact ground state wave-function
is known analytically and that we understand the gap scaling. Such
physical and new models are hard to come by and are valuable for they
teach us new physics of quantum systems. 

This model has unusual properties different from the AKLT and other
such exactly solvable models. On the one hand it has a unique yet
highly entangled ground state, which we nevertheless can analytically
write down and solve (compute entropies, correlations etc.). On the
other hand, in the thermodynamical limit the expected magnetization
in the $z-$direction is a small residue propagated from the boundaries
and is essentially zero in the bulk and the two-point correlation
function in the $z-$direction also vanishes. Moreover, the expected
magnetizations in $x$ and $y$ directions are zero; the model does
not have anisotropies. It seems like the underlying random walks that
describes the ground state are simultaneously responsible for the
long range entanglement and the killing of long range spin correlations. 

Often exactly solvable models such as the AKLT have explicit analytical
expression for the ground state because they are gapped \cite{AKLT}.
Existence of a gap in one-dimension ensures a constant upper bound
on the entropy of entanglement (i.e., rigorously established by the
area law in one-dimension \cite{Matth_areal}). In the more general
case, when the ground state is unique but the gap vanishes in the
thermodynamical limit, it is expected that the area-law conjecture
holds with a possible logarithmic correction. That is as long as the
ground state is unique, and the Hamiltonian is local and translationally
invariant, one expect that the area-law would be violated by at most
a logarithmic factor. This is based mostly on theoretical results
on 1 + 1 dimensional CFTs, as well as, in the Fermi liquid theory
\cite{Korepin2004,holzhey1994geometric,Cardy2009}. The model presented
in \cite{Movassagh2012_brackets} and advanced here, is gapless and
also has a ground state entanglement entropy that scales logarithmically,
but provably does \textit{not} have a CFT in the limit \cite{movassagh2016supercritical}. 

The class of generalized models presented in \cite{movassagh2016supercritical}
is exactly solvable in the same sense as above but are yet much more
surprising. They are integer spin$-s$ quantum spin chains, where
in addition to retaining locality, uniqueness of ground state, and
translational invariance in the bulk, they are very highly entangled:
The half-chain entanglement entropy scales as $\sqrt{n}$ for all
integer spin $s>1$. These models are quite surprising and serve as
counter-examples to the belief that under the constraints imposed
on the interaction and its kernel, logarithmic scaling would be the
maximum violation of the area law. In a future work we will extend
the result herein to compute entanglement and correlation of the generalized
models \cite{movassagh2016supercritical}.\\

There are a number of open problems whose resolution would advance
our understanding of this model and its physics. Such problems include:
\begin{enumerate}
\item Expressing the Hamiltonian of the generalized model presented in \cite{movassagh2016supercritical}
in standard spin representation.
\item Calculation of $\langle s_{n_{1}}^{x}s_{n_{2}}^{x}\rangle$, and $\langle s_{n_{1}}^{x}s_{n_{2}}^{y}\rangle$.
\item Calculation of multipoint correlation functions, e.g., $3-$point
function $\langle{\cal M}_{2n}|s_{n_{1}}^{z}s_{n_{2}}^{z}s_{n_{3}}^{z}|{\cal M}_{2n}\rangle$.
\item Numerical or analytical computation of the time-dependent correlation
functions. Such a $2-$point function is defined by 
\[
\langle{\cal M}_{2n}|s(t)_{n-\frac{L}{2}}^{z}\mbox{ }s(0)_{n+\frac{L}{2}}^{z}|{\cal M}_{2n}\rangle=\langle{\cal M}_{2n}|e^{-iHt}s_{n-\frac{L}{2}}^{z}e^{iHt}\mbox{ }s_{n+\frac{L}{2}}^{z}|{\cal M}_{2n}\rangle.
\]
\item Computation of the dispersion relation, as well as, the low lying
states and energies.
\item The continuum limit of this model.
\item Extending the current results to the case where the boundary projectors
are removed and an external field is added. This modification has
been outlined in \cite{movassagh2016supercritical}.
\end{enumerate}
Any theoretical or numerical work in these directions would be quite
desirable.
\section{Acknowledgements}
I thank Vladimir Korepin for fruitful discussions. I also
thank Adrian Feiguin and Jean-François Marckert for discussions and
Adrian Feiguin for DMRG verifications of the results. I am grateful
for the support and freedom provided by the Herman Goldstine fellowship
in mathematical sciences at IBM TJ Watson Research Center. I thank
the Simons Center for Geometry and Physics for having hosted me during
the workshop of Statistical mechanics and combinatorics (winter of
2016). I thank the Simons Foundation and the American Mathematical
Society for the AMS-Simons travel grant.

\bibliographystyle{plain}
\bibliography{mybib}

\begin{thebibliography}{10}

\bibitem{abgaryan2014quantum}
VS~Abgaryan.
\newblock Quantum entanglement and quantum phase transitions in anisotropic
  two-and three-particle spin-1 heisenberg clusters.
\newblock {\em Journal of Contemporary Physics (Armenian Academy of Sciences)},
  49(6):249--257, 2014.

\bibitem{AKLT}
Ian Affleck, Tom Kennedy, Elliott~H Lieb, and Hal Tasaki.
\newblock Valence bond ground states in isotropic quantum antiferromagnets.
\newblock pages 253--304, 1988.

\bibitem{billingsley2013convergence}
Patrick Billingsley.
\newblock {\em Convergence of probability measures}.
\newblock John Wiley \& Sons, 2013.

\bibitem{Bravyi06}
Sergey Bravyi.
\newblock Efficient algorithm for a quantum analogue of 2-sat.
\newblock {\em Contemporary Mathematics}, 536:33--48, 2011.

\bibitem{Movassagh2012_brackets}
Sergey Bravyi, Libor Caha, Ramis Movassagh, Daniel Nagaj, and Peter~W. Shor.
\newblock Criticality without frustration for quantum spin-1 chains.
\newblock {\em Phys. Rev. Lett.}, 109:207202, 2012.

\bibitem{Bravyi_Gosset2015}
Sergey Bravyi and David Gosset.
\newblock Gapped and gapless phases of frustration-free spin-1/2 chains.
\newblock {\em Journal of Mathematical Physics}, 56:061902, 2015.

\bibitem{Cardy2009}
Pasquale Calabrese and John Cardy.
\newblock Entanglement entropy and conformal field theory.
\newblock {\em J. Phys. A: Math. Theor.}, 42:504005, 2009.

\bibitem{chen2011no}
Jianxin Chen, Xie Chen, Runyao Duan, Zhengfeng Ji, and Bei Zeng.
\newblock No-go theorem for one-way quantum computing on naturally occurring
  two-level systems.
\newblock {\em Physical Review A}, 83(5):050301, 2011.

\bibitem{de1970asymptotic}
Nicolaas~Govert De~Bruijn.
\newblock {\em Asymptotic methods in analysis}, volume~4.
\newblock Courier Corporation, 1970.

\bibitem{facchi2015large}
P~Facchi, G~Florio, G~Parisi, S~Pascazio, and A~Scardicchio.
\newblock Large- n- approximated field theory for multipartite entanglement.
\newblock {\em Physical Review A}, 92(6):062330, 2015.

\bibitem{fan2004entanglement}
Heng Fan, Vladimir Korepin, and Vwani Roychowdhury.
\newblock Entanglement in a valence-bond solid state.
\newblock {\em Physical review letters}, 93(22):227203, 2004.

\bibitem{Naechtaegale1992}
Mark Fannes, Bruno Nachtergaele, and Reinhard~F Werner.
\newblock Finitely correlated states on quantum spin chains.
\newblock {\em Communications in mathematical physics}, 144(3):443--490, 1992.

\bibitem{AdrianChat}
Adrian Feiguin.
\newblock personal communication, 2015.

\bibitem{flajolet2009analytic}
Philippe Flajolet and Robert Sedgewick.
\newblock {\em Analytic combinatorics}.
\newblock cambridge University press, 2009.

\bibitem{gharibian2014quantum}
Sevag Gharibian, Yichen Huang, Zeph Landau, Seung~Woo Shin, et~al.
\newblock Quantum hamiltonian complexity.
\newblock {\em Foundations and Trends{\textregistered} in Theoretical Computer
  Science}, 10(3):159--282, 2015.

\bibitem{gosset2015correlation}
David Gosset, Yichen Huang, et~al.
\newblock Correlation length versus gap in frustration-free systems.
\newblock {\em Physical review letters}, 116(9):097202, 2016.

\bibitem{Matth_areal}
Matthew~B. Hastings.
\newblock An area law for one-dimensional quantum systems.
\newblock {\em Journal of Statistical Physics}, page P08024, 2007.

\bibitem{holzhey1994geometric}
Christoph Holzhey, Finn Larsen, and Frank Wilczek.
\newblock Geometric and renormalized entropy in conformal field theory.
\newblock {\em Nuclear Physics B}, 424(3):443--467, 1994.

\bibitem{kaigh1976invariance}
William~D Kaigh et~al.
\newblock An invariance principle for random walk conditioned by a late return
  to zero.
\newblock {\em The Annals of Probability}, 4(1):115--121, 1976.

\bibitem{Koma95}
Tohru Koma and Bruno Nachtergaele.
\newblock The spectral gap of the ferromagnetic xxz chain.
\newblock {\em Lett. Math. Phys.}, 40:1--16, 1997.

\bibitem{Korepin2004}
V.~E. Korepin.
\newblock Universality of entropy scaling in one dimensional gapless models.
\newblock {\em Phys. Rev. Lett.}, 92:096402, 2004.

\bibitem{korepin2010entanglement}
Vladimir~E Korepin and Ying Xu.
\newblock Entanglement in valence-bond-solid states.
\newblock {\em International Journal of Modern Physics B}, 24(11):1361--1440,
  2010.

\bibitem{Kraus2008}
B.~Kraus, H.~P. B\"uchler, S.~Diehl, A.~Kantian, A.~Micheli, and P.~Zoller.
\newblock Preparation of entangled states by quantum markov processes.
\newblock {\em Phys. Rev. A}, 78:042307, 2008.

\bibitem{landau2013quantum}
Lev~Davidovich Landau and Evgenii~Mikhailovich Lifshitz.
\newblock {\em Quantum mechanics: non-relativistic theory}, volume~3.
\newblock Elsevier, 2013.

\bibitem{li2008entanglement}
Hui Li and F~Duncan~M Haldane.
\newblock Entanglement spectrum as a generalization of entanglement entropy:
  identification of topological order in non-abelian fractional quantum hall
  effect states.
\newblock {\em Physical review letters}, 101(1):010504, 2008.

\bibitem{JK_Chat}
Jean-François Marckert.
\newblock personal communication, 2016.

\bibitem{Movassagh2010PRA}
Ramis Movassagh, Edward Farhi, Jeffrey Goldstone, Daniel Nagaj, Tobias Osborne,
  and Peter~W. Shor.
\newblock Unfrustrated qudit chains and their ground states.
\newblock {\em Phys. Rev. A}, 82:012318, 2010.

\bibitem{movassagh2016supercritical}
Ramis Movassagh and Peter~W Shor.
\newblock Supercritical entanglement in local systems: Counterexample to the
  area law for quantum matter.
\newblock {\em Proceedings of the National Academy of Sciences},
  113(47):201605716, 2016.

\bibitem{cirac}
D.~Perez-Garcia, Frank Verstraete, M.~M. Wolf, and J.~I. Cirac.
\newblock Matrix product representation.
\newblock {\em Quantum Inf. Comput. 7}, 401, 2007.

\bibitem{prokhorov1956convergence}
Yu~V Prokhorov.
\newblock Convergence of random processes and limit theorems in probability
  theory.
\newblock {\em Theory of Probability \& Its Applications}, 1(2):157--214, 1956.

\bibitem{pyrkov2014quantum}
Alexey~N Pyrkov and Tim Byrnes.
\newblock Quantum teleportation of spin coherent states: beyond continuous
  variables teleportation.
\newblock {\em New Journal of Physics}, 16(7):073038, 2014.

\bibitem{sachdev2007quantum}
Subir Sachdev.
\newblock {\em Quantum phase transitions}.
\newblock Wiley Online Library, 2007.

\bibitem{StanleyVol2}
Richard~P. Stanley.
\newblock {\em Enumerative Combinatorics Volume 2}.
\newblock Cambridge University Press, 2001.

\bibitem{Verstraete2009}
Frank Verstraete, Michael~M. Wolf, and J.~Ignacio Cirac.
\newblock Quantum computation and quantum-state engineering driven by
  dissipation.
\newblock {\em Nature Physics}, 5(9):633--636, 2009.

\end{thebibliography}

\end{document}